%% file: recursive_1_.tex
\documentclass[11pt,a4paper,english]{amsart}
\usepackage[utf8]{inputenc}
\usepackage[english]{babel}

\usepackage{amsmath} 
\usepackage{amsthm} 
\usepackage{graphicx} 
\usepackage{xcolor} 
\usepackage{amsfonts}
\usepackage[biblabel]{cite}
\usepackage{bm} 
\usepackage[hidelinks]{hyperref} 
\usepackage{amssymb} 
\usepackage{tikz-cd} 
\usepackage{amscd}
\usepackage{etoolbox}
\patchcmd{\thmhead}{(#3)}{#3}{}{}
\usepackage{enumitem}
\usepackage{breqn} 
\usepackage{faktor} 
\usepackage{xfrac} 
\usepackage{tikz}
\usetikzlibrary{matrix,shapes,arrows,positioning,chains, calc,trees,backgrounds}
\usepackage{comment} 
\usepackage{pdfpages}
\usepackage{xr}

\input{macros.tex}

\usepackage{mathtools} 
\DeclarePairedDelimiter\abs{\lvert}{\rvert}%
\DeclarePairedDelimiter\norm{\lVert}{\rVert}%

\makeatletter
\let\oldabs\abs
\def\abs{\@ifstar{\oldabs}{\oldabs*}}
\let\oldnorm\norm
\def\norm{\@ifstar{\oldnorm}{\oldnorm*}}
\makeatother

\newtheorem{thm}{Theorem}[section]
\newtheorem{prop}[thm]{Proposition}
\newtheorem{cor}[thm]{Corollary}
\newtheorem{lem}[thm]{Lemma}
\theoremstyle{definition}
\newtheorem{defn}[thm]{Definition} 

\newtheorem{rem}[thm]{Remark} 
 
\newtheorem{ex}[thm]{Example}

\title{Structure of weighted projective Reed-Muller codes}

\author[J.~Nardi]{Jade Nardi}
\address[Jade Nardi]{Univ Rennes, CNRS, IRMAR - UMR 6625, Rennes Cedex, France.}
\email{jade.nardi@univ-rennes.fr}

\author[R.~San-José]{Rodrigo San-José}
\address[Rodrigo San-José]{Department of Mathematics\\ Virginia Tech\\ Blacksburg, VA, USA. }
\email{rsanjose@vt.edu}

\subjclass[2020]{94B05, 11T71, 14G50}

\keywords{Weighted projective Reed-Muller codes, recursive construction, generalized Hamming weights, dual code, hull}

\begin{document}

\maketitle

\begin{abstract}
We provide a comprehensive overview of the fundamental structural properties of weighted projective Reed-Muller codes. We give a recursive construction for these codes, under some conditions for the weights, and we use it to derive bounds on the generalized Hamming weights and to obtain a recursive construction for their subfield subcodes and their dual codes. The dual codes are further studied in more generality, where the recursive constructions may not apply, obtaining a description as an evaluation code when the degree is low. We also provide insights into the Schur products of these codes when they are not degenerate.
\end{abstract}

\section{Introduction}
Projective Reed-Muller (PRM) codes were introduced by Lachaud \cite{lachaudPRM}, and their basic parameters were fully determined by S{\o}rensen \cite{sorensen}. They are the projective analogues of affine Reed-Muller (RM) codes \cite{kasamiRM}. With respect to RM codes, PRM codes are longer for the same finite field size, and they have been shown to outperform RM codes when considering the sum of the rate and relative minimum distance \cite{lachaud_parameters_PRM}. A natural generalization of RM codes is given by weighted Reed-Muller (WRM) codes, which are obtained by evaluating polynomials of bounded weighted degree at the rational points of the affine space. These codes were introduced and studied in \cite{sorensenWRM} (also see \cite{olav_WRM_revisited}), where a projective analogue is also mentioned. However, there is a more natural projective extension of WRM codes introduced in \cite{ghorpadeWPRM}, which is the one we will consider in this work. These codes are called \emph{weighted projective Reed-Muller} (WPRM) codes, and they are obtained by evaluating weighted homogeneous polynomials of a fixed degree at the rational points of a weighted projective space. 

The generalized Hamming weights (GHWs) of a linear code form a set of parameters extending the notion of minimum distance. They were introduced by Wei \cite{weiGHW}, which showed that they characterize the performance of a code on the wiretap channel of type II. GHWs have also found other applications over time \cite{guruswammiGHWlistdecoding,matsumotoRGHW,kkks}. For evaluation codes, they admit an interpretation as the maximum number of $\fq$-rational zeros that a system of polynomial equations can have, which is a natural question by itself, and has motivated the study of the GHWs of many different families of evaluation codes \cite{sanjoseGHWNT,pellikaanGHWRM,beelenGHWcartesian,munueraGHWGoppa,munueraGHWhermitica,sanjose_ghw_squarefree}. In particular, Heijnen and Pellikaan computed the GHWs of RM codes \cite{pellikaanGHWRM}, and also mentioned the case of WRM codes. However, the problem of computing the GHWs of PRM codes has been open for more than 25 years. In the introduction of \cite{beelenGHWPRM}, the authors mention many of the different works on the GHWs of PRM codes, and they also propose a conjecture when the degree is lower than the size of the field. For WPRM codes, this problem has not been previously addressed in the literature, and it is worth noting that even the calculation of the minimum distance has proven to be challenging, e.g., see \cite{aubryperretWPRM,nardi_sanjose_conjecture,caki_sahin_nardi25}. 

The study of the duals of evaluation codes is also a classical topic, as they play a crucial role in many different applications, such as decoding algorithms \cite{sakataplusvoting,duursmaMajority,fengraoMajority}, or quantum error-correction \cite{kkks}. The hull of linear codes has also received recent attention due to its use for entanglement-assisted quantum error-correcting codes \cite{galindoentanglement}. The dual of PRM codes was already described in \cite{sorensen}, and their hulls have been studied for some cases \cite{kaplanHullsPRM,sanjoseHullsPRM}. Similarly, the Schur product has found many applications in cryptography \cite{AlainSchurGRSMcEliece1}, multiparty computation \cite{cramer_multiparty_computation}, and quantum fault tolerance \cite{sanjoseCSST}.

In this paper, we study several of the aforementioned properties of linear codes, for WPRM codes. In Section \ref{s:preliminaries} we introduce the necessary preliminaries about weighted projective spaces, their rational points (e.g., see Lemma \ref{l:lambdaQ}), and WPRM codes. In Section \ref{s:recursive}, we introduce a recursive construction for WPRM codes, and we use it to bound their GHWs, and to describe their subfield subcodes and duals. 
In Section \ref{s:duals}, we study the duals from the point of view of evaluation codes, and, for low degree, we provide a generating family formed by the evaluation of a certain set of monomials. Finally, in Section \ref{s:schur}, we leverage the toric geometry of weighted projective spaces to relate the problem of computing the Schur product of two WPRM codes with a question regarding the integer decomposition property of certain simplices.

\section{Preliminaries}\label{s:preliminaries}

\subsection{Linear codes}

Let $\fq$ be the finite field with $q$ elements, where $q$ is a prime power. A \emph{linear code} over $\fq$ is an $\fq$-linear subspace $C\subset \fq^n$. The \emph{dual code} of a linear code $C$, denoted by $C^\perp$, is the orthogonal complement with respect to the usual Euclidean inner product $\ps{\cdot}{\cdot}$, i.e., 
$$
C^\perp :=\{ v\in \fq^n:\ps{c}{v}=0,\; \text{ for all }c\in C\}.
$$
Given two vectors $u,v\in \fq^n$, we denote by $u\star v:=(u_1v_1,\dots,u_nv_n)$ their component-wise product. This is also called sometimes \textit{Schur product} or \textit{star product}. Given two codes $C_1,C_2\subset \fq^n$, we can consider their Schur product $C_1\star C_2:=\langle v_1\star v_2,\; v_1\in C_1,v_2\in C_2 \rangle$. 
We say that two codes $C_1,C_2$ are \textit{monomially equivalent} if there exist $v\in \fq^n$ with nonzero entries, and $\sigma\in S_n$ a permutation, such that $C_2=v \star \sigma(C_1)=\set{v \star \sigma(c_1), \: c_1 \in C_1}$. 

Given a vector $v\in \fq^n$, its \textit{Hamming weight} is the number of nonzero entries of $v$. The \textit{minimum distance} of a linear code is defined as the lowest Hamming weight of a nonzero codeword in $C$. To define GHWs, which were introduced in \cite{weiGHW}, we need the notion of the \emph{support} of a linear subspace $D\subset \fq^n$, which is 
$$
\supp(D):=\left\{ 1\leq i \leq n: \exists \: c \in D \textnormal{ with } c_i \neq 0 \right\}.
$$

\begin{defn} \label{def:ghw} Let $1\leq r \leq k$.
The \emph{$r$-th generalized Hamming weight} (GHW) of an $[n,k,d]$ code $C$ is 
$$
d_r(C):=\min \left\{ \abs{ \supp(D) }: D \textnormal{ is a subcode of } C \textnormal{ of dimension } r \right\}.
$$
The \emph{weight hierarchy} of $C$ is the set $\left\{ d_r(C): 1\leq r \leq k \right\}$.
\end{defn}

\begin{rem}\label{r:ghw_mds}
If $C$ is an $[n,k,d]$ MDS code, i.e., $d=n-k+1$, then we have
$$
d_r(C)=n-k+r, \; 1\leq r \leq k.
$$
This follows from the strict monotonicity of the GHWs \cite[Thm. 1]{weiGHW}.
\end{rem}

Since the computation of the minimum distance is an intractable problem in general \cite{vardyIntractability},  the same holds for the computation of the GHWs of a linear code.

\subsection{Weighted projective spaces}
Let $w=(w_0,w_1,\dots,w_m) \in \NN_{\geq 1}^{m+1}$. The \emph{weighted projective space} (WPS) of weight $w$, denoted by $\PP(w)$, over the field $\F$, is defined as the quotient
\[ \PP(w)=(\A^{m+1}\setminus \{(0,\dots,0)\})/\overline{\F}^*\] under the following action of $\overline{\F}^*$:
$\lambda \cdot (x_0,\dots,x_m)= (\lambda^{w_0} x_0,\dots,\lambda^{w_m}x_m)$ for $\lambda \in \overline{\F}^*$. In the particular case of $w=(1,\dots,1)$, we recover the usual projective space $\PP^m$.  We denote the set of $\fq$-rational points of $\PP(w)$ by $\PP(w)(\fq)$, whose cardinality equals 
\[
\abs{\PP(w)(\fq)}=\frac{q^{m+1}-1}{q-1}=:p_m.
\]
Given an integer $d\geq 0$, we consider $\fq[x_0,\dots,x_m]_d^w$, the vector space of (weighted) homogeneous polynomials of degree $d$, with weight $w$ and coefficients in $\fq$.

\begin{defn}
    Let $w \in \NN_{\geq 1}^{m+1}$. We denote by $\gen{w_0,w_1,\dots,w_m}_\NN$ (or $\gen{w}_\NN$ for short) the semigroup of integers $m$ that can be written as a linear combination of the integers $w_0,w_1,\dots,w_m$ with nonnegative integer coefficients. 
\end{defn}

\begin{defn}\label{def:denum}
    Let $w\in \NN_{\geq 1}^{m+1}$. For any $d \in \NN$, we define the \emph{denumerant} of $d$ with respect to $w$ as
    \[\den(d;w)=\size{\set{(i_0,\dots,i_{m}) \in \NN^{m+1} \text{ such that } w_0 i_0 + \dots+ w_mi_m=d}}.\]
\end{defn}
By definition, $\den(d;w) \geq 1$ if and only if $d \in \gen{w}_\NN$, and $\den(d;w)=\dim \fq[x_0,\dots,x_m]_d^w$. We now give two well-known reductions for the weights of $\PP(w)$.

\begin{lem}\label{l:weightswithgcd}
Let $w= (w_0,\dots,w_m)$ and let $\gamma=\gcd(w_0,\dots,w_m)$. Set $w/\gamma=(w_0/\gamma,\dots,w_m/\gamma)$. Then we have $\PP(w)(\fq)=\PP(w/\gamma)(\fq)$ and  
\[
\fq[x_0,\dots,x_m]_{d}^{w}=\begin{cases} \fq[x_0,\dots,x_m]_{d/\gamma}^{w/\gamma} & \text{if } \gamma \mid d,\\
\{0\} & \text{otherwise.}
\end{cases}
\]
\end{lem}

Due to Lemma \ref{l:weightswithgcd}, we will always assume that $\gcd(w)=1$.

\begin{defn}
A vector of weights $w= (w_0,\dots,w_m)$ is said to be \emph{well-formed} if for every $i \in \set{0,\dots,m}$, $\gcd(w_j, j\neq i)=1$.
\end{defn}

When $w$ is not well-formed, Delorme's reduction \cite{delormeReduction} applies to the WPS $\PP(w)$ and its coordinate ring.

\begin{lem}[(Delorme's weight reduction)]\label{l:delorme}
Let $w= (w_0,\dots,w_m)$. Set $\gamma = \gcd(w_1,\dots,w_m)$. Assume that $\gcd(w_0,\gamma)=1$. The isomorphism 
$$
\begin{array}{lccc}
\varphi: &\PP(w) & \to & \PP(w_0,w_1/\gamma,\dots,w_m/\gamma) \\
& (Q_0:Q_1:\cdots:Q_m) & \mapsto & (Q_0^\gamma:Q_1:\cdots:Q_m).
\end{array}
$$
satisfies
\[
\varphi(\PP(w)(\fq))=\PP(w_0,w_1/\gamma,\dots,w_m/\gamma)(\fq).
\]
Moreover, for any degree $d \geq 0$, we can uniquely write $d=\alpha_0 w_0 + d_0 \gamma$ with $0 \leq \alpha_0 < \gamma$ and
\[
\fq[x_0,\dots,x_m]^w_d=x_0^{\alpha_0} \varphi^* \fq[x_0,\dots,x_m]^{(w_0,w_1/\gamma,\dots,w_m/\gamma)}_{d_0}.
\]
\end{lem}

Consider the following map 
\begin{equation}\label{eq:map_pi}
\begin{array}{lccc}
\pi_w: &\A^{m+1}\setminus \{(0,\dots,0)\} & \to & \PP(w) \\
& (Q_0,\dots,Q_m) & \mapsto & [Q_0:\dots:Q_m].
\end{array}
\end{equation}
By \cite[Prop. 2.1]{aubryperretWPRM}, every $\fq$-point of $\PP(w)$  has a representative in $\mathbb{A}^{m+1}(\fq)\setminus \{0\}$. The next result shows how to obtain all the representatives of a rational point with entries in $\fq$, starting from one such representative. Given a point $Q=(Q_0,\dots,Q_m)$, we set $\supp(Q):=\{i \in \set{0,\dots,m} \mid Q_i\neq 0\}$. 

\begin{lem}\label{l:lambdaQ}
Let $w=(w_0,\dots,w_m) \in \NN^{m+1}$. Let $Q=Q^{(1)}=(Q_0,\dots,Q_m)$ be a representative for a rational point in $\PP(w)(\fq)$, and denote by $Q^{(2)},\dots, Q^{(q-1)}$ its other $q-2$ representatives (there are $q-1$ in total). Let $\xi$ be a primitive element of $\fq$, and consider $\lambda \in \overline{\fq}$ a root of $x^{\gcd(w_i:i\in \supp(Q))}-\xi$. 
Then
\begin{equation}\label{eq:severalRep}
\{ \lambda^i\cdot Q^{(1)},\;0\leq i \leq q-2\}= \{Q^{(1)},\dots, Q^{(q-1)}\}.
\end{equation}
\end{lem}
\begin{proof}
First note that $\supp(Q)=\supp(Q^{(i)})$ for any $i$, which means that $\gcd(w_i:i\in \supp(Q))$ depends only on the rational point. We have $(\lambda^{iw_0},\dots,\lambda^{iw_m})\in (\fq^*)^{m+1}$, since, for any $j$, there is some $t$ such that $\lambda^{iw_j(q-1)}=\xi^{t i(q-1)}=1$. Moreover, $\lambda^i\cdot Q^{(1)}=\lambda^j\cdot Q^{(1)}$ with $i<j$ implies $\lambda^{(j-i)w_\ell}=1$, for any $\ell \in \supp(Q^{(1)})$. If we consider Bezout's identity $\sum_{\ell\in \supp(Q^{(1)})} u_\ell w_\ell=\gcd(w_i:i\in \supp(Q))$, then
$$
1=\prod_{\ell\in \supp(Q^{(1)})} \lambda^{(j-i)u_\ell w_\ell}=\lambda^{(j-i)\gcd(w_i:i\in \supp(Q))}=\xi^{j-i},
$$
a contradiction since $j-i<q-1$ and $\ord(\xi)=q-1$. 

\end{proof}

\begin{rem}
Note that, in practice, we can construct the representatives $\{Q^{(1)},\dots,Q^{(q-1)}\}$ from the previous result without considering any field extension, since we have that $\lambda^{w_i}=\xi^{w_i/\gcd(w_i:i\in \supp(Q))} \in \fq$.
\end{rem}

Lemma \ref{l:lambdaQ} can be considered as a refinement of \cite[Lemma 7]{perretNumberPointsWeightedProjectiveSpace}, since it gives a constructive way to obtain the representatives. This result, together with \cite[Prop. 2.1]{aubryperretWPRM}, provides a direct proof of the fact that $\abs{\PP(w)(\fq)}=p_m$. Indeed, from the proof of Lemma \ref{l:lambdaQ} we have that any $\fq$-rational point $[Q]\in \PP(w)(\fq)$ has exactly $q-1$ $\fq$-representatives, i.e., $\abs{\pi_w^{-1}([Q])}=q-1$ (recall Equation~\eqref{eq:map_pi}). Since these preimages are disjoint, it follows that $\abs{\PP(w)(\fq)}=\frac{q^{m+1}-1}{q-1}=p_m$. The usual way to prove this relies on Hilbert's Theorem 90. This alternative approach uses \cite[Prop. 2.1]{aubryperretWPRM} instead, and it is constructive: given an $\fq$-point, the proof of \cite[Prop. 2.1]{aubryperretWPRM} shows how to obtain one representative with coordinates in $\fq$, and Lemma \ref{l:lambdaQ} gives a way to get all the other representatives. 

\begin{rem}\label{r:lambda_in_fq}
For any representative $Q^{(1)}$ such that $\gcd(w_i:i\in \supp(Q^{(1)}))=1$, we have just shown that we may consider $\lambda \in \fq^*$ in Lemma \ref{l:lambdaQ}. Indeed, if $\lambda^{\gcd(w_i:i\in \supp(Q^{(1)}))}=\xi$, let $\nu $ be such that $\nu \gcd(w_i:i\in \supp(Q^{(1)}))\equiv 1\bmod q-1$, and then we may choose $\lambda=\xi^\nu$ (see also the proof of \cite[Lem. 7]{perretNumberPointsWeightedProjectiveSpace}). In particular, if $\gcd(w_i,q-1)=1$, for $0\leq i \leq m$, we may choose $\lambda\in \fq^*$ for any point of $\PP(w)(\fq)$.     
\end{rem}
As the next example shows, some cases require $\lambda\in \overline{\fq}\setminus \fq^*$. 
\begin{ex}
Let $q=3$, and $w=(2,3)$. If we consider $Q^{(1)}=(1,0)$, then, according to Lemma \ref{l:lambdaQ}, we need to consider a root of $x^2-(-1)=x^2+1$. Since there is no root for that polynomial in $\F_3$, we deduce that $\lambda \in \overline{\F_3}\setminus \F_3$. We have $\lambda^2=-1$ and $\lambda\cdot Q^{(1)}=(-1,0)=Q^{(2)}$. 
\end{ex}

\subsection{Weighted projective Reed-Muller codes}

Fixing $\mathcal{P}_w=(P_1,\dots,P_{p_m})$ a set of representatives for $\PP(w)(\fq)$, we can define an evaluation map
\begin{equation}\label{eq:ev_Pm}
\begin{array}{lccc}
\ev_{\mathcal{P}_w}: & \fq[x_0,\dots,x_m]_d^w & \to & \fq^n,\\
& f & \mapsto & (f(P_1),\dots,f(P_{p_m})).
\end{array}
\end{equation}
Note that $\ev_{\mathcal{P}_w}$ depends on $d$, but we will not make this dependence explicit for ease of notation.

\begin{defn}
The \textit{weighted projective Reed-Muller} code is the linear code $\WPRM_d(w):=\ev_{\mathcal{P}_w}(\fq[x_0,\dots,x_m]_d^w)$. If $w=(1,\dots,1)$, we recover \textit{projective Reed-Muller} codes, which are denoted $\PRM_d(m)$. 
\end{defn}
The previous definition depends on the choice of representatives $\mathcal{P}_w$, but different choices give monomially equivalent codes, as the next result shows.

\begin{lem}\label{l:change_of_representatives}
Let $d>0$ and consider $\mathcal{P}_w,\mathcal{P}'_w$ two sets of representatives for the points of $\PP(w)(\fq)$ such that for every $Q\in \mathcal{P}_w$, we have $\lambda_Q\cdot Q\in \mathcal{P}'_w$ (as in Lemma \ref{l:lambdaQ}). Assuming the same order for the points, we have
$$
\ev_{\mathcal{P}'_w}(\fq[x_0,\dots,x_m]^w_d)=(\lambda_Q^{d})_{Q\in \calP_w}\star\ev_{\mathcal{P}_w}(\fq[x_0,\dots,x_m]^w_d).
$$
\end{lem}
\begin{proof}
It follows from the fact that, for $g\in \fq[x_0,\dots,x_m]_d^w$, we have $g( \lambda_Q\cdot Q)=\lambda_Q^{d}g(Q).$
\end{proof}
\begin{cor}
Let $d>0$ such that $\gcd(d,q-1)=1$. Assume $\gcd(w_i,q-1)=1$, for $0\leq i \leq m$. Then every code that is monomially equivalent to $\WPRM_d(w)$ can be seen as a WPRM code of degree $d$, evaluating at a different set of representatives. 
\end{cor}
\begin{proof}
The result holds if and only if, for every $Q\in \mathcal{P}_w$, we have $\{\lambda_Q^{id}:0\leq i \leq q-2\}=\fq^*$, which happens if and only if the order of $\lambda_Q^d$ is $q-1$. By Remark \ref{r:lambda_in_fq}, we may choose $\lambda_Q\in \fq^*$, and it will have order $q-1$ by Lemma \ref{l:lambdaQ}. Note that, for permutations, we may just choose a different order for the points of $\PP(w)(\fq)$.
\end{proof}

Unlike WRM and PRM codes, WPRM codes may be degenerate in some cases. We can characterize precisely when this happens. 

\begin{lem}\label{l:degen}
$\WPRM_d(w)$ is nondegenerate if and only if $\lcm(w)\mid d$. 
\end{lem}
\begin{proof}
If $\lcm(w)\nmid d$, then $w_i \nmid d$, for some $0\leq i \leq m$. Thus, no monomial of the form $x_i^\alpha$ has weighted degree $d$. This implies that the point $[0:\dots:0:1:0:\dots:0]$, with a single 1 in position $i$, is a common zero of all the homogeneous polynomials of degree $d$.

Conversely, if $\lcm(w)\mid d$, for each $0\leq i \leq m$, we have that $x_i^{d/w_i}$ is of weighted degree $d$. A common zero of these monomials would have to have the $i$th coordinate equal to 0, for all $0\leq i \leq m$, and there is no such point in $\PP(w)$. 
\end{proof}

We can rephrase Lemmas~\ref{l:weightswithgcd} and \ref{l:delorme} in terms of codes.

\begin{cor}\label{c:codes_weightswithgcd}
Let $w= (w_0,\dots,w_m)$ and let $\gamma=\gcd(w_0,\dots,w_m)$. Set $w/\gamma=(w_0/\gamma,\dots,w_m/\gamma)$. Then for any degree $d \geq 0$,
\[
\WPRM_d(w)=\begin{cases}
    \WPRM_{d/\gamma}(w/\gamma) & \text{if } \gamma \mid d,\\
    \set{\mathbf{0}_{p_m}} & \text{otherwise.}
\end{cases}
\]
\end{cor}

\begin{cor}\label{c:codes_delorme}
Let $w= (w_0,\dots,w_m)$. Set $\gamma = \gcd(w_1,\dots,w_m)$. Assume that $\gcd(w_0,\gamma)=1$. For any degree $d \geq 0$, we can uniquely write $d=\alpha_0 w_0 + d_0 \gamma$ with $0 \leq \alpha_0 < \gamma$ and
\[
\WPRM_d(w)=\ev_{\calP_w}(x_0^{\alpha_0}) \star
    \WPRM_{d_0}(w_0,w_1/\gamma,\dots,w_m/\gamma)
\]
where $\calP_w$ is the set of representatives of $\PP(w)(\fq)$ used to define the left-hand side code, and $\varphi(\calP_w)=\set{\varphi(Q), Q \in \calP_w}$ is the one for the right-hand side code.
\end{cor}

\begin{proof}
    From Lemma~\ref{l:delorme}, the set $\varphi(\calP_w)$ forms a set of representatives of $\PP(w_0,w_1/\gamma,\dots,w_m/\gamma)(\fq)$. Moreover, any polynomial $f \in \fq[x_0,\dots,x_m]^w_d$ can be written uniquely as $f=x_0^{\alpha_0} g(x_0^\gamma,x_1,\dots,x_m)$ for some $g \in \fq[x_0,\dots,x_m]^{(w_0,w_1/\gamma,\dots,w_m/\gamma)}_{d_0}$. One can easily check, by definition of the pullback, that $\ev_{\calP_w}(f)=\ev_{\calP_w}(x_0^{\alpha_0})\star \ev_{\varphi(\calP_w)}(g)$.
\end{proof}

Note that the kernel of the evaluation map $\ev_{\mathcal{P}_m}$ does not depend on the choice of $\mathcal{P}_m$. Let $\calI(\PP(w)(\fq))$ be the ideal generated by the homogeneous polynomials that vanish at all the $\fq$-points of $\PP(w)$. Then we have 
$$
\fq[x_0,\dots,x_m]_d^w/\calI(\PP(w)(\fq))_d \cong \WPRM_d(w). 
$$
From \cite[Thm. 3.5]{nardi_projective_toric} and \cite{mesutComputingVanishing}, we have the following result about $\calI(\PP(w)(\fq))$.

\begin{thm}\label{t:binomials_vanishing_ideal}
The ideal $\calI(\PP(w)(\fq))$ is binomial. Moreover, a homogeneous binomial $x^\alpha-x^\beta$ lies in $\calI(\PP(w)(\fq)) $ if and only if $\alpha_i=0 \iff \beta_i=0$,  and $q-1\mid \beta_i-\alpha_i$, for $0\leq i \leq m$.
\end{thm}


We will also use the affine counterpart of WPRM codes. We denote by $\fq[x_1,\dots,x_m]^w_{\leq d}$ the polynomials of (weighted) degree less than or equal to $d$. If we enumerate $\mathbb{A}^m(\fq)=\{P_1,\dots,P_{q^m}\}$, we can consider the evaluation map 
$$
\begin{array}{lccc}
\ev_{\mathbb{A}^m}: & \fq[x_1,\dots,x_m] & \to & \fq^n,\\
& f & \mapsto & (f(P_1),\dots,f(P_{q^m})).
\end{array}
$$
\begin{defn}
The \textit{affine weighted Reed Muller} code is the linear code $\WRM_d(w):=\ev_{\mathbb{A}^m}(\fq[x_1,\dots,x_m]^w_{\leq d})$. 
\end{defn}

Let $w':=(w_1,\dots,w_m)$. We will also consider the following subcode of $\WRM_d(w)$: 
\begin{equation}\label{eq:defWRMweird}
\WRM_{d}(w_0;w'):=\set{\ev_{\A^m}(x^\alpha): \alpha\in \NN^m,\; \sum_{i=1}^m w_i\alpha_i\leq d, \; \sum_{i=1}^m w_i\alpha_i\equiv d\bmod w_0}.
\end{equation}

If $\gamma=\gcd(w)\mid d$, then it follows from the definitions that $\WRM_{d}(w_0;w')=\WRM_{d/\gamma}(w_0/\gamma;w'/\gamma)$. We can also obtain a result similar to Corollary \ref{c:codes_delorme} for these codes. 

\begin{lem}\label{l:delorme_weird_WRM}
Let $w= (w_0,\dots,w_m)$ and $w'=(w_1,\dots,w_m)$. Set $\gamma = \gcd(w')$. Assume that $\gcd(w_0,\gamma)=1$. For any degree $d \geq 0$, we can uniquely write $d=\alpha_0 w_0 + d_0 \gamma$ with $0 \leq \alpha_0 < \gamma$ and
\[
\WRM_d(w_0;w')= \WRM_{d_0}(w_0;w_1/\gamma,\dots,w_m/\gamma).
\]
\end{lem}
\begin{proof}
It is clear that $ \sum_{i=1}^m w_i\alpha_i\equiv d\bmod w_0$ if and only if $ \sum_{i=1}^m (w_i/\gamma)\alpha_i \equiv d_0\bmod w_0$. We also have that $\sum_{i=1}^m (w_i/\gamma)\alpha_i\leq d_0$ implies $\sum_{i=1}^m w_i \alpha_i\leq d_0\gamma\leq d$, which proves $\WRM_d(w_0;w')\supset \WRM_{d_0}(w_0;w_1/\gamma,\dots,w_m/\gamma)$. Moreover, if $\sum_{i=1}^m w_i\alpha_i\leq d$, then 
$$
\sum_{i=1}^m \frac{w_i\alpha_i}{\gamma}\leq \frac{\alpha_0w_0}{\gamma}+d_0.
$$
As $\frac{\alpha_0w_0}{\gamma} < w_0$, we get $\frac{\alpha_0w_0}{\gamma}+d_0<w_0+d_0$. If $\sum_{i=1}^m (w_i/\gamma) \alpha_i\equiv d_0 \bmod w_0$, we cannot get $\sum_{i=1}^m (w_i/\gamma) \alpha_i= d_0 +j$ for any $1\leq j <w_0$. Therefore, the conditions $\sum_{i=1}^m w_i\alpha_i\leq d$ and $\sum_{i=1}^m w_i\alpha_i\equiv d\bmod w_0$ imply $\sum_{i=1}^m (w_i/\gamma) \alpha_i\leq d_0$, which proves the reversed inclusion. 
\end{proof}

Note that $\WRM_d(w_0;w')\subset \WRM_d(1;w')=\WRM_{d}(w')$. For the next result, recall the vanishing ideal of the set of $\fq$-points of the affine space $\A^m$:
\begin{equation}\label{eq:vanishing_ideal_affine}
\calI(\mathbb{A}^m(\fq))=\langle x_i^q-x_i,\; 1\leq i \leq m \rangle.     
\end{equation}
Therefore $x^\alpha\equiv x^\beta \bmod \calI(\mathbb{A}^m(\fq))$ if and only if, for each $1\leq i \leq m$, we have $\alpha_i=0$ if and only if $\beta_i=0$, and $\alpha_i\equiv \beta_i\bmod q-1$.

\begin{lem}\label{l:dim_weird_rm}
Let $w=(w_0,\dots,w_m)\in \NN_{\geq 1}^{m+1}$ and $w'=(w_1,\dots,w_m)$. Let $1\leq d\leq w_0(q-1)$. If $\gcd(w_0,q-1)=1$, then
$$
\dim \WRM_d(w_0;w')=\den(d;w).
$$
\end{lem}
\begin{proof}
The map from $\{(\ell_0,\dots,\ell_m)\in \NN^{m+1}:\sum_{i=0}^m \ell_iw_i=d \}$ to 
\[
A=\{ (\ell_{1},\dots,\ell_m)\in \NN^{m}: \sum_{i=1}^m \ell_i w_i \equiv d\bmod w_0,\;\sum_{i=1}^m \ell_i w_i\leq d \}
\]
that sends $(\ell_0,\dots,\ell_m)$ to $(\ell_{1},\dots,\ell_m)$ is a bijection between those two sets. Note that the cardinality of the first set is $\den(d;w)$, and the cardinality of $A$ is equal to the number of monomials we evaluate to construct $\WRM_d(w_0;w')$ in \eqref{eq:defWRMweird}. Thus, we only need to prove that, given $x^\alpha,x^\beta$ with distinct $\alpha,\beta \in A$, we cannot have $x^\alpha \equiv x^\beta\bmod \calI(\mathbb{A}^m(\fq))$. If we had $x^\alpha \equiv x^\beta\bmod \calI(\mathbb{A}^m(\fq))$, then $\alpha_i\equiv \beta_i\bmod q-1$, for $1\leq i \leq m$, which implies $\sum_{i=1}^m \alpha_iw_i\equiv \sum_{i=1}^m \beta _iw_i \bmod q-1$. Since $\alpha,\beta \in A$, we have $\sum_{i=1}^m \alpha_iw_i\equiv \sum_{i=1}^m \beta _iw_i \bmod w_0$. Taking into account that $\gcd(w_0,q-1)=1$, we obtain $\sum_{i=1}^m \alpha_iw_i\equiv \sum_{i=1}^m \beta _iw_i \bmod w_0(q-1)$. As $\alpha \neq \beta$, then $\sum_{i=1}^m \alpha_iw_i \neq \sum_{i=1}^m \beta _iw_i$, which contradicts the fact that both of these sums must be smaller than or equal to $d\leq w_0(q-1)$.
\end{proof}

For an $[n,k]$ linear code $C\subset\fq^n$, define $\Delta(C)=\frac{k+d_1(C)}{n}$. We use this parameter to show that WPRM codes can outperform WRM codes, similarly to what is done in \cite{lachaud_parameters_PRM} for PRM and RM codes. For ease of comparison, let $w_0=1$, $w_1=\min(w')$, and $w_1\mid d$. Assume that $d<q$ to ensure
\[
\den(d;w)=\dim \WRM_d(w')=\dim \WPRM_d(w).
\]
Then we have  $\Delta(\WRM_d(w'))<\Delta(\WPRM_d(w))$ if and only if
$$
p_m(\den(d;w)+(q-d/w_1)q^{m-1})<q^m(\den(d;w)+(q-d/w_1+1)q^{m-1})
$$
(see \cite{sorensenWRM} and \cite{nardi_sanjose_conjecture} for the minimum distance of these codes), which can be translated to
$$
p_{m-1}\den(d;w)<q^{m-1}(q^m-p_{m-1}(q-d/w_1))=q^{m-1}(1+p_{m-1}(d/w_1-1)).
$$
The latter holds when $\den(d;w)\leq q^{m-1}(d/w_1-1)$. Since the left-hand side does not depend on $q$, then the inequality holds for large enough $q$.




\subsection{Weighted projective Reed-Solomon codes}\label{ss:wprs}
The case $m=1$ corresponds to weighted projective Reed-Solomon (WPRS) codes. In that case, we can directly determine all the parameters. Notice that we may always assume $\gcd(w_0,w_1)=1$ by Lemma \ref{l:weightswithgcd}. We use the notation $\RS_\delta(X)$ for the Reed-Solomon (RS) code obtained by evaluating the monomials $\{1,x,\dots,x^\delta\}$ at the points of $X\subset \fq$, and $\WPRS_\delta(w_0,w_1):=\WPRM_\delta(w_0,w_1)$.  We also denote $\PRS_\delta:=\WPRS_\delta(1,1)$. 

\begin{prop}\label{prop:WPRS}
Let $(w_0,w_1) \in \NN^2$ with $\gcd(w_0,w_1)=1$ and $d\geq 1$.
Set 
\begin{equation}\label{eq:delta}
    \delta=\den(d;w_0,w_1)-1.
\end{equation}

\begin{itemize}
    \item If $w_0w_1 \mid d$, then $\WPRS_d(w_0,w_1)=\PRS_{\delta}$.
    \item If either $w_0$ or $w_1$ divides $d$ (but not both), then $\WPRS_d(w_0,w_1)$ is monomially equivalent to $\set{0} \times \RS_\delta(\fq)$.
    \item If neither $w_0$ nor $w_1$ divides $d$, then $\WPRS_d(w_0,w_1)$ is monomially equivalent to $\set{(0,0)} \times \RS_\delta(\fq^*)$. 
\end{itemize}

\end{prop}

\begin{proof}
It follows from the proof of \cite[Prop. 4.1]{nardi_sanjose_conjecture}.
\end{proof}

\begin{cor}\label{cor:dist_WPRS}
    Let $(w_0,w_1) \in \NN^2$ with $\gcd(w_0,w_1)=1$ and $d \geq 0$. Set $\rho$ the remainder of the Euclidean division of $d$ by $w_0w_1$, i.e., $d \equiv \rho \bmod (w_0w_1)$ with $0 \leq \rho < w_0w_1$.
    The minimum distance of $\WPRS_d(w_0,w_1)$ is equal to
    \[d_1(\WPRS_d(w_0,w_1))=\max\set{q-\Floorfrac{d-1}{w_0w_1}-\tilde\epsilon,1}\]
    where $\tilde\epsilon=\begin{cases}
        \den(\rho;w_0,w_1) & \text{if } w_0 \nmid d \text{ and } w_1 \nmid d,\\
        0 &\text{otherwise.}
    \end{cases}$
\end{cor}

\begin{proof} It follows from \cite[Cor. 4.3]{nardi_sanjose_conjecture}.
\end{proof}
Since RS and PRS codes are MDS, we can derive the rest of the GHWs of WPRS codes using Proposition \ref{prop:WPRS} and Remark \ref{r:ghw_mds}.

\section{Recursive construction of WPRM codes}\label{s:recursive}

Let $w=(w_0,\dots,w_m) \in \NN^{m+1}$, and assume $\gcd(w_0,q-1)=1$. Then, from \cite[Lemma 3.1]{aubryperretWPRM} we have 
\begin{equation}\label{eq:1xA}
\mathbb{P}(w_0,\dots,w_m)(\fq)=\left[\{1\}\times \A^m(\fq) \right]\cup \{0\}\times \mathbb{P}(w_1,\dots,w_m)(\fq).
\end{equation}

\begin{defn}
The \textit{affine cone} associated to $X\subset \PP(w)$ is defined as
\begin{equation}\label{eq:cone}
\Cone(X):=\pi_w^{-1}(X)\cup \{(0,\dots,0)\}. 
\end{equation}    
\end{defn}
Let $w'=(w_1,\dots,w_m)$. One can check (for example, see \cite[Cor. 2.33]{ghorpadeWPRM}) that
$$
\Cone(\PP(w')(\fq))=\A^m(\fq).
$$
In particular, this implies that 
\begin{equation}\label{eq:AmPw}
\A^m(\fq)\setminus \{(0,\dots, 0)\}=\bigcup_{i=1}^{q-1} \calP_{w'}^i,
\end{equation}
where $\calP_{w'}^1, \dots, \calP_{w'}^{q-1}$ are disjoint sets of representatives for $\PP(w')(\fq)$. Using both Equations ~\eqref{eq:1xA} and \eqref{eq:AmPw}, we can choose the following ordered set of representatives $\calP_w$ for the $\fq$-points on $\PP(w)$:
\begin{equation}\label{eq:nice_rep}
\begin{aligned}
\calP_w = &\bigsqcup_{i=1}^{q-1} \set{(1,y_1,\dots,y_m),\: (y_1,\dots,y_m) \in \calP^i_{w'}} \\
&\sqcup \set{(1,0,\dots,0)} \sqcup \set{(0,y_1,\dots,y_m),\: (y_1,\dots,y_m) \in \calP^1_{w'}}.
\end{aligned}
\end{equation}
To make the recursive construction more explicit, we choose the disjoint sets of representatives of $\PP(w')(\fq)$ as follows. Take $\calP^1_{w'}$ be a set of representatives of $\PP(w')(\fq)$. For every $Q \in \calP^1_{w'}$, we fix $\lambda_Q$ as in Lemma~\ref{l:lambdaQ} and for $2 \leq i \leq q-1$, we set
\begin{equation}\label{eq:lambdaQ}
\calP_{w'}^i:=\{ \lambda_Q^{i-1}\cdot Q:Q\in \calP^1_{w'}\}.
\end{equation}

For the case $w=(1,\dots,1)$, it is always possible to obtain the sets $\calP_{w'}^2,\dots,\calP_{w'}^{q-1}$ with the form $\calP_{w'}^i=\{ \lambda^{i}\cdot Q:Q\in \calP^1_{w'}\}$ for a common $\lambda$ (which is a primitive element of $\fq^*$ in this case). This feature is used for the recursive construction of PRM codes, see \cite{sanjoseRecursivePRM}. In the next example, we show that this may not be possible for general weights.

\begin{ex}\label{ex:P23}
In $\PP(2,3)(\F_3)$, we have $-1 \cdot (1,1)=(1,-1)$ so the points $[1:1]$ and $[1:-1]$ are equal. Let us consider the set of representatives
$$
\calP_{(2,3)}=\{(1,0),(0,1),(1,1),(-1,-1)\}.
$$
Any other set of representatives $\calP_{(2,3)}'$ contains both $(-1,0)$ and $(0,-1)$. However, if we assume that there is $\lambda \in \overline{\fq}$ such that $\calP_{(2,3)}'=\lambda \cdot \calP_{(2,3)}$, we get
$$
\begin{aligned}
&\lambda \cdot (1,0)=(\lambda^2,0)=(-1,0)\iff \lambda^2=-1, \\
&\lambda \cdot (0,1)=(0,\lambda^3)=(0,-1)\iff\lambda^3=-1.
\end{aligned}
$$
This is a contradiction, since this implies $-1=\lambda^3=\lambda^2\lambda=-\lambda$, but $1^3=1\neq -1 = \lambda^3$. For any other starting set of representatives $\calP_{(2,3)}$, an analogous argument shows that one cannot obtain another disjoint set of representatives in this manner.
\end{ex}

We use the notation $(u,v)$ to denote the concatenation of two vectors $u,v$, and we use the notation $\mathbf{0}_n$ to denote the zero vector of length $n$, whenever ambiguity may rise.

\begin{thm}\label{t:rec1weakcond}
Let $w=(w_0,\dots,w_m) \in \NN^{m+1}$ with $\gcd(w_0,q-1)=1$. Set $w'=(w_1,\dots,w_m)$. Let $\calP_w$ be the fixed ordered set of representatives of $\PP(w)(\fq)$ defined in Equation~\eqref{eq:nice_rep} with $\calP^1_{w'}$ a set of representatives of $\PP(w')(\fq)$, and $\calP^2_{w'},\dots,\calP^{q-1}_{w'}$ as in Equation~\eqref{eq:lambdaQ}. Set
\begin{equation}\label{eq:Lambda(i)}
\Lambda(i):=\left(\lambda_Q^{(i-1)d}\right)_{Q\in \calP^1_{w'}} \in (\fq^*)^{p_{m-1}}  . 
\end{equation}
Then
$$
\begin{aligned}
\WPRM_d(w)=\{ (u+v_\Lambda,v):\; &u\in \WRM_{d-w_0}(w_0;w'),v\in \WPRM_d(w') \},
\end{aligned}
$$
where $v_\Lambda:=v\times \Lambda(2) \star  v\times \cdots \times \Lambda(q-1) \star v\times \{0\}=(v,\Lambda(2) \star  v, \Lambda(3) \star  v,\dots,\Lambda(q-1) \star v,\mathbf{0}_1)$. 
\end{thm}
\begin{proof}
For any $f\in \fq[x_0,\dots,x_m]_d^w$,
\begin{equation}\label{eq:decomposition}
f=x_0f'+g
\end{equation}
with $f'\in \fq[x_0,\dots,x_m]_{d-w_0}^{w}$ and $g\in \fq[x_1,\dots,x_m]_d^{w'}$. Given the choice of the representatives $\calP_w$ (see Equation~\eqref{eq:nice_rep}), we have
$$
\ev_{\calP_w}(x_0f')=(u, \mathbf{0}_{p_{m-1}}),
$$
where $u\in \WRM_{d-w_0}(w_0;w')$. Indeed, $f'$ is homogeneous of degree $d-w_0$, and its evaluation on $\{1\}\times \A^m(\fq)$ is the same as the evaluation of $f'(1,x_1,\dots,x_m)$ in $\A^m(\fq)$. This is a polynomial of weighted degree lower than or equal to $d-w_0$, and all of the monomials in its support have degree equivalent to $d$ modulo $w_0$. 

On the other hand, considering Equations~\eqref{eq:nice_rep} and \eqref{eq:lambdaQ} together with Lemma \ref{l:change_of_representatives}, we obtain
$$
\ev_{\calP_w}(g)=v_{\Lambda}\times v=(v_\Lambda,v),
$$
where $v\in \WPRM_d(w')$. This concludes the proof of the fact that, given $f\in \fq[x_0,\dots,x_m]_d^w$, its evaluation in $\calP_w$ is of the form $(u+v_\Lambda,v)$, with $u\in \WRM_{d-w_0}(w_0;w'),v\in \WPRM_d(w')$. 

Reciprocally, let us prove now that any vector of that form is in $\WPRM_d(w)$. Set $u\in \WRM_{d-w_0}(w_0;w')$, and let $f'$ be the polynomial with weighted degree lower than or equal to $d-w_0$ whose evaluation in $\A^m(\fq)$ is $u$. Since all the monomials in the support of $f'$ have degree equivalent to $d$ modulo $w_0$, this polynomial can be homogenized with the variable $x_0$ to a homogeneous polynomial of degree $d$, and its evaluation in $\calP_w$ is the vector $(u, \mathbf{0}_{p_{m-1}})$, which is thus in $\WPRM_d(w)$. Given $v\in \WPRM_d(w')$, there is a homogeneous polynomial $g$ in the variables $x_1,\dots,x_m$ whose evaluation in $\calP_{w'}$ is $v$. Arguing as above, the evaluation of $g$ in $\calP_w$ is $(v_{\Lambda},v)$, and this vector is in $\WPRM_d(w)$. 
\end{proof}

\begin{cor}\label{cor:delorme_not_divisible}
Consider the setting from Corollary \ref{c:codes_delorme} and assume also that $\gcd(w_0,q-1)=1$. If $\gamma\nmid d$, then 
$$
\WPRM_d(w)=\WRM_{d-w_0}(w_0;w')\times \{ {\bf 0 }_{p_{m-1
}}\}=\WRM_{d_0}(w_0;w'/\gamma)\times \{ {\bf 0 }_{p_{m-1
}}\}.
$$  
\end{cor}
\begin{proof}
By Lemma~\ref{l:delorme}, any $f \in \fq[x_0,\dots,x_m]^w_d$ can be written as $f=x_0^{\alpha_0}h$ with $h \in \fq[x_0,\dots,x_m]^{(w_0,w_1/\gamma,\dots,w_m/\gamma)}_{d_0}$. Note that $\alpha_0>0$. In Equation~\eqref{eq:decomposition}, we obtain $g=0$, which gives the first equality. Since $\gcd(w_0,q-1)=1$, we also have $\ev_{\A^m}(f)=\ev_{\A^m}(f')=\ev_{\A^m}(h)$, which gives the second equality. 
\end{proof}

If there are only two weights $w_i,w_j$ such that $\gcd(w_i,q-1)\neq 1$, $\gcd(w_j,q-1)\neq 1$, without loss of generality we may assume that $i=m-1$ and $j=m$. Then we may apply Theorem~\ref{t:rec1weakcond} to get a complete recursive construction, which will eventually involve $\WPRM_d(w_{m-1},w_m)$, whose structure and parameters are fully known (see Subsection \ref{ss:wprs}). Otherwise, Theorem~\ref{t:rec1weakcond} cannot be used to its full extent, but we can leverage Corollaries~\ref{c:codes_weightswithgcd} and \ref{c:codes_delorme} to improve the applicability of the recursive construction to the component codes, as we show in the next example. We also give an example in which we cannot say anything with Theorem \ref{t:rec1weakcond}.  

\begin{ex}\label{ex:tree}
Let us apply the recursive construction to $\WPRM_d(1,2,3,6)$ (see Figure~\ref{fig:tree} for a visual summary). By Theorem~\ref{t:rec1weakcond}, we can construct $\WPRM_d(1,2,3,6)$ with $\WRM_{d-1}(1;2,3,6)=\WRM_{d-1}(2,3,6)$ and $\WPRM_d(2,3,6)$. Now $w'=(2,3,6)$ is not well-formed and we can apply Corollary~\ref{c:codes_delorme} twice: writing $d=2 \alpha_1 + 3 \alpha_2 + 6 d_1$ with $0 \leq \alpha_1 <3$ and $0 \leq \alpha_2 <2$, we get 
\[
\begin{aligned}
\WPRM_d(2,3,6)&=\ev_{\calP_{w'}}(x_1^{\alpha_1})\star \WPRM_{\alpha_2+2d_1}(2,1,2)\\
&= \ev_{\calP_{w'}}(x_1^{\alpha_1}) \star \ev_{\varphi(\calP_{w'})}(x_2^{\alpha_2})\star \PRM_{d_1}(2)\\
&= \ev_{\calP_{w'}}(x_1^{\alpha_1} x_2^{\alpha_2})\star \PRM_{d_1}(2),
\end{aligned}
\]
where
\[
\begin{array}{lccc}
\varphi: &\PP(2,3,6) & \to & \PP(2,1,2) \\
& (Q_0:Q_1:Q_2) & \mapsto & (Q_0^3:Q_1:Q_2).
\end{array}
\]
We can then apply Theorem~\ref{t:rec1weakcond} (or equivalently \cite[Theorem~3.1]{sanjoseRecursivePRM}) to $\PRM_{d_1}(2)$ and construct it from $\RM_{d_1-1}(2)$ and $\PRS_{d_1}$.
If $\gcd(2,q-1)=1$ (resp., $\gcd(3,q-1)=1$), then we can also apply Theorem~\ref{t:rec1weakcond} to $\WPRM_d(2,3,6)$ (resp., $\WPRM_d(3,2,6)$, which is the same code), to construct this code from $\WRM_{d-2}(2;3,6)$ and $\WPRS_d(3,6)$ (resp., $\WRM_{d-3}(3;2,6)$ and $\WPRS_d(2,6)$). 
\end{ex}

\begin{ex}
    Let $q=31$ and $w=(2,3,5)$. We cannot apply any weight reduction, and all the weights have nontrivial greatest common divisor with $q-1=30=\lcm(w)$, which forbids any use of Theorem~\ref{t:rec1weakcond}.
\end{ex}

\begin{figure}
    \centering
    \includegraphics[width=\textwidth]{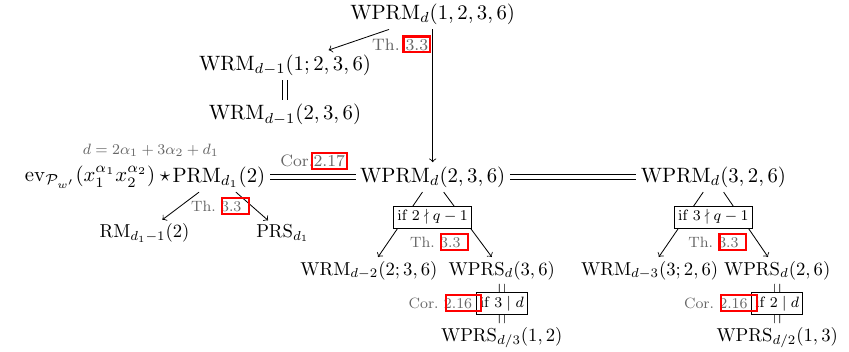}
    
    \caption{Tree of possible weight reductions (via Corollaries \ref{c:codes_weightswithgcd} and \ref{c:codes_delorme}), and decompositions (via Theorem \ref{t:rec1weakcond}) corresponding to Example \ref{ex:tree}.}
    \label{fig:tree}
\end{figure}

In what follows, we will derive properties for $\WPRM_d(w)$ using Theorem \ref{t:rec1weakcond}, for the first step of the recursion. Arguing as in Examples \ref{ex:tree}, in many cases we will be able to keep using the recursive construction until only codes with known parameters remain. 

\begin{cor}\label{c:dimension}
Let $w=(w_0,\dots,w_m) \in \NN^{m+1}$ with $\gcd(w_0,q-1)=1$, and $w'=(w_1,\dots,w_m)$. Then
$$
\dim \WPRM_d(w)=\dim \WRM_{d-w_0}(w_0;w')+\dim \WPRM_d(w').
$$
\end{cor}
\begin{proof}
This is a direct consequence of Theorem \ref{t:rec1weakcond}, taking into account that the vectors of the form $(v_\Lambda,v)$ and the vectors of the form $(u,\mathbf{0}_{p_{m-1}})$ are linearly independent.  
\end{proof}

Given a subfield $\F_{q'}$ of $\fq$, and a code $C\subset \fq^n$, its subfield subcode with respect to this field extension is $C_{q'}:=C\cap \F_{q'}^n$. This is a well-known technique to obtain long codes over smaller finite field sizes, and many families of codes with good parameters can be obtained in this way \cite{bierbrauercyclic,sanjoseSSCPRM,sanjoseSSCPRS,goppaOriginal}, and provide good candidates for the McEliece cryptosystem \cite{nardiSSCAG,couvreurSSCMcEliece,McEliece1978}. We have the following result on the subfield subcodes of weighted projective Reed-Muller codes.

\begin{cor}\label{c:ssc}
Let $w=(w_0,\dots,w_m) \in \NN^{m+1}$, $w'=(w_1,\dots,w_m)$ and $\F_{q'}\subset \fq$. If $d=\ell \frac{(q-1)\lcm(w)}{q'-1}$, for some $\ell \geq 1$, then, with the notation as in Theorem \ref{t:rec1weakcond}, we have
$$
(\WPRM_d(w))_{q'}=\{ (u+v_\Lambda,v):\; u\in (\WRM_{d-w_0}(w_0;w'))_{q'},v\in (\WPRM_d(w'))_{q'} \}. 
$$
As a consequence,
$$
\dim (\WPRM_d(w))_{q'}=\dim (\WRM_{d-w_0}(w_0;w'))_{q'}+\dim (\WPRM_d(w'))_{q'}.
$$
\end{cor}
\begin{proof}
First, we show that $\Lambda(i)\in \F_{q'}^{p_m}$, for $1\leq i \leq q-1$. Following the setting from Lemma \ref{l:lambdaQ}, for each point $Q\in \PP(w)(\fq)$, we choose $\lambda_Q$ as a root of $x^{\gcd(w_i:i\in \supp(Q))}-\xi$, where $\xi$ is a primitive element of $\fq$. Then 
$$
\lambda_Q^{(q'-1)(i-1)d}=\lambda_Q^{\ell(i-1)(q-1)\lcm(w)}=1,
$$
which implies that $\lambda_Q^{(i-1)d}\in \F_{q'}$, i.e., $\Lambda(i)\in \F_{q'}^{p_m}$. If $u$ and $v$ have their entries in $\F_{q'}$, it is clear that $(u+v_\Lambda,v)\in \F_{q'}^{p_m}$. Reciprocally, if $(u+v_\Lambda,v)\in \F_{q'}^{p_m}$, then $v\in \F_{q'}^{p_{m-1}}$, which implies $v_\Lambda\in \F_{q'}^{q^m}$ and $u\in \F_{q'}^{q^m}$. The statement about the dimension follows as in Corollary \ref{c:dimension}.
\end{proof}

To apply the previous result recursively, we also need to understand the subfield subcodes of weighted Reed-Muller codes. These can be seen as a particular case of J-affine variety codes, for which we have bounds for the minimum distance and formulas for the dimension \cite{galindostabilizer}.

With respect to the minimum distance and the GHWs of WPRM codes, we have an analogous result to \cite[Thm. 7]{sanjoseRecursivePRM} or \cite[Thm. 3.1]{sanjoseGHWMPC}. By convention, for the following result we will define $d_0(C)=0$, $d_r(C)=\infty$ if $r>\dim(C)$, and the $r$-th GHW of the zero code is defined to be $0$, for any $r$. We also consider that $\WPRM_0(w)=\set{ \mathbf{0}_{p_m}
}$. 

\begin{thm}\label{t:boundghw}
Let $d\geq 1$, $1\leq r\leq \dim(\WPRM_d(w))$, $w=(w_0,\dots,w_m) \in \NN^{m+1}$ with $\gcd(w_0,q-1)=1$, and $w'=(w_1,\dots,w_m)$. Set
\[
E=\begin{cases}
\WRM_{d-(q-1)\max\{w_0,\min(w')\}}(w')&\text{if } d>(q-1)\max\{w_0,\min(w')\},\\
\set{ \mathbf{0}_{q^m}}& \text{otherwise.} 
\end{cases}
\]
Let $R:=\{0,\dots,r\}\times \{0,\dots,r\}$, and 
$$
    Y = \left\{(\alpha_1,\alpha_2) \in R: \begin{array}{c}
    r-\dim\WRM_d(w_0;w') \leq  \alpha_1 \leq  \dim E\\
     r-\dim \WPRM_d(w') \leq \alpha_2 \leq  \dim \WRM_{d-w_0}(w_0;w')  \\
        \alpha_1+\alpha_2\leq r
        \end{array}\right\}.
$$
Then
$$
d_r(\WPRM_d(w))\geq \min_{(\alpha_1,\alpha_2)\in Y}B_{\alpha_1,\alpha_2},
$$
where 
$$
\begin{aligned}
B_{\alpha_1,\alpha_2}:=& \max\{d_{r-\alpha_1}(\WRM_d(w_0;w')),d_{\alpha_2}(\WRM_{d-w_0}(w_0;w'))\} \\
&+\max\set{\left\lceil \frac{d_{\alpha_1}\left(E\right)}{q-1}\right\rceil,d_{r-\alpha_2}(\WPRM_d(w'))}.
\end{aligned}
$$
\end{thm}
\begin{proof}
Let $D\subset \WPRM_d(w)$ be a subcode with $\dim D=r$. Following the notation from Theorem \ref{t:rec1weakcond}, we assume that  $u\in \WRM_{d-w_0}(w_0;w')$ and $ v\in \WPRM_d(w')$ in what follows. We define
$$
\begin{aligned}
&D_1:=\{(u+v_\Lambda, v)\in D: u+v_\Lambda=\mathbf{0}_{q^m}\},\\
&D_2:=\{(u+v_\Lambda, v)\in D: v=\mathbf{0}_{p_{m-1}}\}.
\end{aligned}
$$
We also consider $D_3$ such that $D_1\oplus D_2\oplus D_3=D$. 

Let $\alpha_i:=\dim D_i$, $1\leq i \leq 2$. Note that $(u+v_\Lambda, v)\in D_3\setminus\set{\mathbf{0}_{p_m}}$ if and only if $u+v_\Lambda\neq \mathbf{0}_{q^m}$ and $v\neq \mathbf{0}_{p_m-1}$. Using the decomposition from Equation~(\ref{eq:1xA}), we split $\supp(D)$ into 
\begin{equation}\label{eq:split}
    \supp(D)= \supp_{\text{aff}}(D) \sqcup \supp_\infty(D)
\end{equation}
with
\[\begin{aligned}
    \supp_{\text{aff}}(D)&:=\supp(D)\cap\{1,\dots,q^m\},\\
    \supp_\infty(D)&:=\supp(D)\cap \{q^m+1,\dots,p_m\}.
\end{aligned}\]
First, we bound $\abs{\supp_{\text{aff}}(D)}$. On one hand, note that 
\begin{equation}\label{eq:aff1}
\abs{\supp_{\text{aff}}(D)}=\abs{\supp_{\text{aff}}(D_2\oplus D_3)}\geq \abs{\supp(D_2)}\geq d_{\alpha_2}(\WRM_{d-w_0}(w_0;w')), 
\end{equation}
since $D_2\subset \WRM_{d-w_0}(w_0;w')\times \set{\mathbf{0}_{p_{m-1}}}$. On the other hand, any vector $(u+v_\Lambda,v)\in \WPRM_d(w)$ satisfies that $u\in \WRM_{d-w_0}(w_0;w')\subset \WRM_d(w_0;w')$ and that $v_\Lambda$ is the evaluation of a homogeneous polynomial of degree $d$ in the variables $x_1,\dots,x_m$ at $\A^m(\fq)$, i.e., $v_\Lambda\in \WRM_d(w_0;w')$. This implies that 
\begin{equation}\label{eq:aff2}
\abs{\supp_{\text{aff}}(D)}=\abs{\supp_{\text{aff}}(D_2\oplus D_3)}\geq d_{r-\alpha_1}(\WRM_d(w_0;w')), 
\end{equation}
since $\dim(D_2\oplus D_3)= r-\alpha_1$. Gathering \eqref{eq:aff1} and \eqref{eq:aff2}, we conclude that 
$$
\abs{\supp_{\text{aff}}(D)}\geq \max\{d_{r-\alpha_1}(\WRM_d(w_0;w')),d_{\alpha_2}(\WRM_{d-w_0}(w_0;w'))\}.
$$
With respect to $\abs{\supp_{\infty}(D)}$, we have
$$
\abs{\supp_{\infty}(D)}  =\abs{\supp_{\infty}(D_1\oplus D_3)}  \geq d_{r-\alpha_2}(\WPRM_d(w')),
$$
because $\dim(D_1\oplus D_3)=r-\alpha_2$ and the last $p_m-q^m=p_{m-1}$ coordinates of any vector in $\WPRM_d(w)$ belong to $\WPRM_d(w')$. On the other hand, we also have
\begin{equation}\label{eq:lowerbound_D1}
\abs{\supp_{\infty}(D)}=\abs{\supp_{\infty}(D_1\oplus D_3)}\geq \abs{\supp(D_1)}\geq \left\lceil \frac{d_{\alpha_1}\left(E\right)}{q-1}\right\rceil.
\end{equation}
We only need to prove the last inequality.
Let $f \in \fq[x_0,\dots,x_m]^w_d$ such that $\ev_{\calP^m}(f)=(u+v_\Lambda,v) \in D_1$. Set $f'$ and $g$ as in the proof of Theorem \ref{t:rec1weakcond}, and $f''=f'(1,x_1,\dots,x_m)$. 
Then $\ev_{\A^m}(f'')=u$, $\ev_{\A^m}(g)=v_\Lambda$ and the condition $u+ v_\Lambda=0$ implies that
\begin{equation}\label{eq:reductionIA}
f''\equiv -g\bmod \langle x_1^q-x_1,\dots,x_m^q-x_m\rangle.
\end{equation}
Let $\overline{f''}$ and $\overline{g}$ be the polynomials obtained from $f''$ and $g$, respectively, where all the monomials have their exponents reduced modulo $q-1$. Then $\overline{f''}=-\overline{g}$. Since $f''$ is a polynomial of degree at most $d-{w_0}$, we obtain that $\overline{g}$ is of degree at most $d-{w_0}<d$, and since $g$ is homogeneous of degree $d$, then all the monomials in $g$ can be reduced modulo $\langle x_1^q-x_1,\dots,x_m^q-x_m\rangle$, and they have degree $\geq 1$. 
This means that the degree of $\overline{g}$ is, at most, $d-(q-1)\min(w')$. However, we can be more precise. The degree of a monomial from $\overline{f''}$ can be written as $d-\lambda_0w_0-(q-1)\sum_{i\geq 1}\lambda_i w_i$, for some $\lambda_i\geq 0$, $\lambda_0\geq 1$. Analogously, a monomial from $\overline{g}$ has degree $d-(q-1)\sum_{i\geq 1}\mu_iw_i$, for some $\mu_i\geq 0$, where not all $\mu_i$ are zero. Since $\overline{f''}=-\overline{g}$, these degrees agree and we have
$$
\lambda_0w_0+(q-1)\sum_{i\geq 1}\lambda_iw_i=(q-1)\sum_{i\geq 1}\mu_iw_i.
$$
Taking into account that $\gcd(w_0,q-1)=1$, we get $(q-1)\mid \lambda_0$. If $d-(q-1)\max\{w_0,\min(w')\}\geq 1$, as $\lambda_0 \geq 1$, we actually have $\lambda_0 \geq q-1$ and $1\leq \deg(\overline{g})\leq d-(q-1)\max\{w_0,\min(w')\}$. Thus, $v_\Lambda=\ev_{\A^m}(g)=\ev_{\A^m}(\overline{g})\in \WRM_{d-(q-1)\max\{w_0,\min(w')\}}(w')=E$. 

On the other hand, now assume $d-(q-1)\max\{w_0,\min(w')\}\leq 0$. Since  $1\leq \deg(\overline{g})$ if $g\neq 0$, the only possible $\overline{g}$ satisfying the previous conditions is $\overline{g}=0=g$, and therefore $v_\Lambda\in \set{\mathbf{0}_{q^m}} = E$. Finally, since $v_\Lambda=(v,\Lambda(2) \star  v, \Lambda(3) \star  v,\dots,\Lambda(q-1) \star v,\mathbf{0}_1)$, we have $\abs{\supp(v_\Lambda)}=(q-1)\abs{\supp(v)}$, which proves the last inequality in Equation~\eqref{eq:lowerbound_D1}. 
We complete the proof by noticing that $\abs{\supp(D)}=\abs{\supp_{\text{aff}}(D)}+\abs{\supp_{\infty}(D)}$ (see Equation~\eqref{eq:split}), and $(\alpha_1,\alpha_2)\in Y$.  
\end{proof}

Let us write the previous theorem for the minimum distance, i.e., for $r=1$.

\begin{cor}\label{c:boundminimumdistance}
Let $d\geq 1$, $w=(w_0,\dots,w_m) \in \NN^{m+1}$ with $\gcd(w_0,q-1)=1$, and $w'=(w_1,\dots,w_m)$. Then
$$
\begin{aligned}
d_1(\WPRM_d(w))\geq \min \{ &d_1(\WRM_{d-w_0}(w_0;w')),d_1(\WRM_{d-(q-1)\max\{w_0,\min(w')\}}(w')),\\
&d_1(\WRM_d(w_0;w'))+\wt(\WPRM_d(w')) \}.
\end{aligned}
$$
\end{cor}
\begin{proof}
In this case, we have $Y=\{(0,0),(1,0),(0,1)\}$ and we apply Theorem \ref{t:boundghw} for $r=1$, unless $\dim \WRM_{d-w_0}(w_0;w')=0$ or $\dim \WRM_d(w_0;w')=0$. In those cases, $Y$ is a subset of the aforementioned one, and, due to our conventions, the formula still holds. 
\end{proof}

As we will see in Example~\ref{ex:boundGHWsdifferentW}, the bound of Theorem~\ref{t:boundghw} may depend on the ordering of the weights (if there are several weights $w_j$ with $\gcd(w_j,q-1)=1$). Thus, for each possible $r$, we may take the maximum over the values of the bound for all the possible orderings of the weights. However, one can easily check that if the vector of weight is not well-formed, i.e., there exists $i_0 \in \set{0,\dots,m}$ such that $\gamma=\gcd(w_i, i\neq i_0) >1$ with $\gcd(\gamma,w_{i_0})=1$, and if $\gamma$ divides $d$, then the bound of Theorem~\ref{t:boundghw} is either equal or sharper when applied to the reduced code (see Corollary~\ref{c:codes_delorme}) compared to the bounds for the original code.

Note that the codes appearing in Theorem \ref{t:boundghw} and Corollary \ref{c:boundminimumdistance} are either WPRM codes, for which we may be able to apply the results again recursively, or WRM codes (with the standard definition, or that from (\ref{eq:defWRMweird})). WRM codes can be understood as decreasing cartesian codes, and thus the footprint bound gives their GHWs (this can be proven in a similar way to the proof given in \cite{eduardoGHWHyperbolic} for hyperbolic codes). For the codes from (\ref{eq:defWRMweird}), the footprint bound is not necessarily sharp, but we can still use it as a bound for their GHWs (it is equivalent to using $d_r(\WRM_{d-w_0}(w_0;w'))\geq d_r(\WRM_{d-w_0}(w'))$). Clearly, if we substitute $\WRM_{d-w_0}(w_0;w')$ with $\WRM_{d-w_0}(w')$ in the definition of $B_{\alpha_1,\alpha_2}$, we still get a lower bound for $d_r(\WPRM_d(w))$ as in Theorem \ref{t:boundghw}. Now we provide upper bounds for the GHWs of WPRM codes to complement the previous results.

\begin{lem}\label{l:upperboundghw}
Let $d\geq 1$, $w=(w_0,\dots,w_m) \in \NN^{m+1}$ with $\gcd(w_0,q-1)=1$, $w'=(w_1,\dots,w_m)$ and $1\leq r \leq \max \{\dim \WRM_{d-w_0}(w_0;w'), \dim \WPRM_d(w'))\}$. Then
$$
d_r(\WPRM_d(w))\leq \min\{d_r(\WRM_{d-w_0}(w_0;w')),q\hspace{0.05cm}d_r(\WPRM_d(w'))\}.
$$
\end{lem}
\begin{proof}
We use the notation from Theorem \ref{t:rec1weakcond}. If $r\leq \dim \WRM_{d-w_0}(w_0;w') $, we can find a subcode $D$ of $\WPRM_d(w)$ with $\dim D=r$ and $D$ is generated by vectors of the type $(u_i,0)$, with $u_i\in \WRM_{d-w_0}(w_0;w')$, for $i=1,\dots,r$. If we assume that the cardinality of the support of the code generated by $\{u_i\}_{i=1}^r\subset \WRM_{d-w_0}(w_0;w'))$ is $d_r(\WRM_{d-w_0}(w_0;w'))$, we obtain $d_r(\WPRM_d(w))\leq \abs{\supp(D)}=d_r(\WRM_{d-w_0}(w_0;w'))$. 

If $r\leq \dim \WPRM_d(w') $, similarly we may also find a subcode $D$ with $\dim D=r$ which is generated by vectors of the type $(v_\Lambda,v)$, and such that $d_r(\WPRM_d(w))\leq \abs{\supp(D)}=q\hspace{0.05cm}d_r(\WPRM_d(w'))$.
\end{proof}

For $d\leq \min(w)q$, we can obtain an upper bound similar to that in \cite[Thm. 2.3]{beelenGHWPRM}. For the following result, we denote $w(a):=(w_a,\dots,w_m)$, for any $0\leq a \leq m$. 

\begin{prop}\label{p:upper_bound_conjecture}
Let $w=(w_0,\dots,w_m) \in \NN_{\geq 1}^{m+1}$ and let $1\leq d \leq  \min(w)q$. Let $1\leq r \leq \den(d;w)$, and let $0\leq i \leq m+1$ and $0\leq j < \den(d-w_i;w(i))$ be the unique integers such that
$$
r=\sum_{a=0}^{i-1}\den(d-w_a;w(a))+j.
$$
Then, if $\gcd(w_i,q-1)=1$, we have
$$
d_r(\WPRM_d(w))\leq q^{m-i+1}p_{i-1}+d_j(\WRM_{d-w_i}(w_i;w(i+1))),
$$
with the convention $p_{-1}=0$.
\end{prop}
\begin{proof}
For $0\leq a \leq m$, we denote by $B_a$ a basis for $x_a\fq[x_a,\dots,x_m]_{d-w_a}^{w(a)}$. By construction, $\abs{B_a}=\den(d-w_a; w(a))$. Since $d\leq \min(w)q$, we have that $B_a$ is also linearly independent modulo $\calI(\PP(w(a))(\fq))$, for $0\leq a \leq m$. In particular, as the union of the $B_a$'s for $0\leq a \leq m$ forms a basis of $\fq[x_a,\dots,x_m]_{d}^{w}$, we get that $\dim \WPRM_d(w)=\den(d;w)=\sum_{a=0}^m \den(d-w_a; w(a))$. Thus, we may always write
$$
r=\sum_{a=0}^{i-1}\den(d-w_a;w(a))+j,
$$
for some $0\leq i \leq m+1$, $0\leq j < \den(d-w_i;w(i))$. Note that this implies $j<\dim \WRM_{d-w_i}(w_i;w(i+1))$ by Lemma \ref{l:dim_weird_rm}. Therefore, there exist $j$ polynomials 
\[f_1,\dots,f_j\in\Span\set{ x^\alpha: \alpha\in \NN^{m-i+1},\; \sum_{k=i+1}^m w_k\alpha_k\leq d-w_i, \; \sum_{k=i+1}^m w_k\alpha_k\equiv d-w_i\bmod w_i}\] such that $\abs{V_{\mathbb{A}^{m-i}}(f_1,\dots,f_j)(\fq)}=q^{m-i}-d_j(\WRM_{d-w_i}(w_i; w(i+1)))$, where $V_{\mathbb{A}^{m-i}}(f_1,\dots,f_j)(\fq)$ denotes the common zeroes of $f_1,\dots,f_j$ in $\mathbb{A}^{m-i}(\fq)$. We denote by $F_1,\dots,F_j \in \fq[x_i,\dots,x_m]^{w(i)}_{d}$ the homogenization of these polynomials to degree $d$, using the variable $x_i$. 
The set $B=\left(\bigcup_{a=0}^{i-1} B_a \right)\cup \{F_1,\dots,F_j\}$ has cardinality $r$. Let $[Q_0:\dots:Q_m]\in V_{\PP(w)}(B)(\fq)=V_{\PP(w)}\left(\bigcup_{a=0}^{i-1} B_a \right)(\fq)\cap V_{\PP(w)}(F_1,\dots,F_j)(\fq)$. Note that
\begin{equation}\label{eq:v_of_b}
V_{\PP(w)}\left(\bigcup_{a=0}^{i-1} B_a \right) \supset  V_{\PP(w)}(x_0,x_1,\dots,x_{i-1})= \{ [Q_0:\dots :Q_m] :Q_0=\dots=Q_{i-1}=0 \}.
\end{equation}

Assume that $Q_0=\dots=Q_{i-1}=0$. \begin{itemize}
    \item Either $Q_i=0$ (note that $x_i$ divides the polynomials $F_1,\dots,F_j$), 
    \item or $Q_i\neq 0$, and then $(Q_{i+1},\dots,Q_m)\in V_{\mathbb{A}^{m-i}}(f_1,\dots,f_j)(\fq)$.
\end{itemize} Thus,
\[\begin{aligned}
    \abs{V_{\PP(w)}(B)(\fq)}& \geq \abs{V_{\PP(w)}(B)(\fq)\cap V_{\PP(w)}(x_0,x_1,\dots,x_{i-1})}\\
    &=p_{m-i-1} + \left( q^{m-i} - d_j(\WRM_{d-w_i}(w_i;w(i+1)))\right)\\
    & = p_{m-i}-d_j(\WRM_{d-w_i}(w_i;w(i+1)),
\end{aligned}\]
 which implies 
\[\begin{aligned}
    d_r(\WPRM_d(w))&\leq p_m - (p_{m-i}-d_j(\WRM_{d-w_i}(w_i;w(i+1)))\\
    &= q^{m-i+1}p_{i-1}+d_j(\WRM_{d-w_i}(w_i;w(i+1))).
\end{aligned}\]

\end{proof}

\begin{rem}
In Proposition \ref{p:upper_bound_conjecture}, if $w_0,\dots,w_{i-1}\mid d$, i.e., $\lcm(w_0,\dots,w_{i-1})\mid d$, we have equality in Equation \eqref{eq:v_of_b}, and the bound given for $\abs{V(B)}$ is an equality. 
Moreover, having $d\leq \min(w)q$ is a necessary condition to ensure $\fq[x_0,\dots,x_m]_d^w \cap \calI(\PP(w)(\fq))=\set{0}$. The latter can hold for higher degrees depending on the structure of the numerical semigroup $\langle w_0,\dots,w_m\rangle _{\NN}$, see \cite{mesutComputingVanishing}.
\end{rem}

The bound from Proposition \ref{p:upper_bound_conjecture} is conjectured to be sharp when $1\leq d <\min(w)(q-1)$ and $w=(1,\dots,1)$ in \cite{beelenGHWPRM}. For $r=1$ and $\min(w)=1$, we know the exact value of $d_1(\WPRM_d(w))$ \cite[Theorem~1.2]{nardi_sanjose_conjecture}. For higher values of $r$, and $w=(1,\dots,1)$, we know some partial results, e.g., see \cite{beelenGHWPRM}. As illustrated in \cite[\textsection 4]{nardi_sanjose_conjecture}, it seems difficult to state an explicit conjecture in the case $\min(w)>1$. This is corroborated by Proposition~\ref{p:upper_bound_conjecture}:  the bound for $d_1(\WPRM_d(w))$ involves $d_1(\WRM_{d-w_i}(w_i,w(i+1)))$, for which we do not know a closed formula if $w_i>1$. If $\min(w)=1$ (or if we have some weight that is coprime with $q-1$), we can particularize Proposition \ref{p:upper_bound_conjecture} and obtain the following result. 

\begin{cor}
Let $w=(w_0,\dots,w_m) \in \NN_{\geq 1}^{m+1}$ and let $1\leq d \leq  \min(w)q$. Let $1\leq r < \den(d-w_0;w)$. Then, if $\gcd(w_0,q-1)=1$, we have
$$
d_r(\WPRM_d(w))\leq d_r(\WRM_{d-w_0}(w_0 ; w')).
$$
\end{cor}



\subsection{Examples}
This is the first time that the GHWs of WPRM codes are studied (besides the case $w=(1,\dots,1)$), and thus we can only compare our bounds with the true value of the GHWs and not with other bounds. Since the computation of GHWs is NP-hard \cite{vardyIntractability}, we will restrict ourselves to small examples, and we will use the Sage \cite{sagemath} implementation given in \cite{sanjoseGHWsPackage,githubGHWs} to obtain the true values of the GWHs.

\begin{ex}\label{ex:boundGHWscomplete}

Let $q=3$, $w=(3,1,1)$ and $d=3$. Now we may apply Theorem \ref{t:boundghw}. We give below the parameters of the codes involved, in the format $[n,k,(d_1(C),\dots,d_k(C))]$, and the set $Y$. The GHWs of $\WPRS$ codes are known due to Remark \ref{r:ghw_mds} and Corollary \ref{cor:dist_WPRS}, and the GHWs of the subcodes of WRM codes defined in Equation~\eqref{eq:defWRMweird} have been directly computed with \cite{sanjoseGHWsPackage,githubGHWs}. One could also use the footprint bound to estimate the GHWs of the WRM codes from (\ref{eq:defWRMweird}), but since we want to test the sharpness of the bound from Theorem \ref{t:boundghw}, and not the tightness of the footprint bound for those WRM codes, we use the true value of the GHWs, either using \cite{githubGHWs} or Remark \ref{r:ghw_mds} if the corresponding codes are MDS. Note that in this case $d-(q-1)\max\{w_0,\min(w')\}<0$, and thus we do not need to consider the corresponding code for the bound.

\begin{table}[ht]
\caption{Parameters of the constituent codes from Theorem \ref{t:boundghw}.}
\label{table:exGHWs}
\centering

\begin{tabular}{c|c}
Code $C$ & $[n,k,(d_1(C),\dots,d_k(C))]$ \\
\hline
$\WRM_{3}(3;(1,1))$ & $[9,5,(2,4,6,8,9)]$ \\
$\WRM_{0}(3;(1,1))=\Span((1,\dots,1))$ &$[9,1,(9)]$\\
$\WPRM_{3}(1,1)=\PRS_3$ & $[4,4,(1,2,3,4)]$\\
\end{tabular}
\end{table}
From these parameters, we obtain 
\[
Y=\begin{cases}
    \{(0,0),(0,1)\} &\text{for }1\leq r \leq 3,\\
    \{(0,1)\} &\text{for } r=4.
\end{cases}
\]
For example, for $r=3$, we can compute
$$
\begin{aligned}
&B_{0,0}=\max\{ 6,0\}+\max\{0,3\}=9,\\
&B_{0,1}=\max\{6,9\}+\max\{0,2\}=11.
\end{aligned}
$$
Thus,
$$
d_3(\WPRM_3(3,1,1))\geq \min\{9,11\}=9. 
$$
If we consider now Lemma \ref{l:upperboundghw}, we obtain
$$
d_3(\WPRM_3(3,1,1)) \leq 3 d_3(\WPRM_{3}(1,1))=d_3(\PRS_3)=9.
$$
Therefore, we have obtained $d_3(\WPRM_3(3,1,1))=9$. Similarly, one can check that the bound from Theorem \ref{t:boundghw} is sharp for $1\leq r \leq 5=\dim \WPRM_3(3,1,1)$, obtaining the weight hierarchy $(3,6,9,12,13)$.

With $d=6$, the bounds of Theorem \ref{t:boundghw} match the weight hierarchy  $(2,3,5,6,8,9,11,12,13)$. By Corollary~\ref{c:codes_delorme}, $\WPRM_6(3,1,1)=\WPRM_{12}(3,2,2)$. With this representation, the values given by Theorem~\ref{t:boundghw} are $(1,2,3,4,6,8,10,12,13)$, which illustrates the advantage of using Delorme's reduction.
\end{ex}

\begin{ex}\label{ex:boundGHWsdifferentW}
Let $w=(2,3,5)$, $q=4$, and $d=30=\lcm(2,3,5)$. This is the first degree for which the code is nondegenerate (see Lemma \ref{l:degen}). The bound from Theorem \ref{t:boundghw} is sharp in this case if we consider the order $w=(2,3,5)$, and it gives the weight hierarchy $(2, 3, 4, 5, 6, 7, 8, 9, 10, 11, 12, 14, 15, 16, 19, 20, 21)$. 

Now consider $q=3$ and $d=20$. In this case, the code is degenerate by Lemma \ref{l:degen}, and the cardinality of its support is 12 instead of 13. If we compute the bound from Theorem \ref{t:boundghw} for the ordering of the weights $w=(3,2,5)$, we obtain the values $(2,3,4,5,6,8,9,10,12)$. If we compute them with the ordering $w=(5,2,3)$, we obtain $(1,2,3,5,6,7,9,11,12)$ instead. For each $r$, we may take the maximum of the values we have obtained, and thus we obtain $(2,3,4,5,6,8,9,11,12)$, which is, in fact, the weight hierarchy of $\WPRM_{20}(2,3,5)$. This shows the benefit of using several orderings for the weights, and it also shows that, in principle, there is no single best ordering of the weights for the bound, since in this case the first ordering gives a better bound for $r=6$, but a worse bound for $r=8$, with respect to the second ordering. The rest of the orderings of the weights give either the values of $w=(3,2,5)$ or the values of $w=(5,2,3)$, since the only weight that matters for the bound is $w_0$ if we are using \cite{sanjoseGHWsPackage} for the GHWs of the constituent codes (if we were using the bound recursively, then the ordering of the last weights could also be relevant). 
\end{ex}

\section{Duals}\label{s:duals}
In this section, we study the duals of WPRM codes. The case $m=1$ follows from what is known for RS and PRS codes, using Proposition \ref{prop:WPRS}. 
\begin{prop}\label{prop:dual_WPRS}
Let $(w_0,w_1) \in \NN^2$ with $\gcd(w_0,w_1)=1$ and $d \geq 0$.
Set 
\begin{equation}
    \delta=\den(d;w_0,w_1)-1.
\end{equation} 
\begin{enumerate}
    \item If $w_0w_1 \mid d$, then $\WPRS^\perp_d(w_0,w_1)=\PRS_{q-1-\delta}$.
    \item If either $w_0$ or $w_1$ divides $d$ (but not both), then $\WPRS_d(w_0,w_1)^\perp$ is monomially equivalent to $\set{0} \times \RS_{q-1-\delta}(\fq)+\langle (1,0,\dots,0) \rangle$.
    \item If neither $w_0$ nor $w_1$ divides $d$, then $\WPRS_d(w_0,w_1)^\perp$ is monomially equivalent to $\set{(0,0)} \times \RS_{q-1-\delta}(\fq^*)+\langle (1,0,0\dots,0),(0,1,0,\dots,0)  \rangle$. 
\end{enumerate}
\end{prop}

In the next result, we show that the duals can be constructed recursively. Note that the duals of WRM codes are also WRM codes \cite{sorensenWRM}, and, thus, we know their parameters. With respect to the WRM codes introduced in (\ref{eq:defWRMweird}), their duals can be understood in the context of J-affine variety codes, see \cite[Prop. 1 \& 2]{galindostabilizer}. 

\begin{prop}\label{p:dual}
Assume the setting from Theorem~\ref{t:rec1weakcond}. For $u^t\in \WRM_{d-w_0}^\perp(w_0;w')$, we write $u^t=(u_1^t,\dots,u_{q-1}^t,u_q^t)$, where $u_i^t\in \fq^{p_{m-1}}$, for $1\leq i \leq q-1$, is the vector formed by the coordinates of $u^t$ corresponding to $\calP_w^i$, and $u_q^t\in \fq$ corresponds to the point $(0,\dots,0)$. Then
$$
\WPRM_d^\perp(w) =\{ (u^t,v^t-u^t_\Lambda):\; u^t\in \WRM_{d-w_0}^\perp(w_0;w'),v^t\in \WPRM_d^\perp(w') \},
$$
where $u_\Lambda^t:=\sum_{i=1}^{q-1} \Lambda(i)\star u_i^t$. 
\end{prop}
\begin{proof}
By Corollary~\ref{c:dimension}, the given vector space has the same dimension as $\WPRM_d^\perp(w)$. We only need to prove that it is orthogonal to $\WPRM_d(w)$. Let $u \in \WRM_{d-w_0}(w_0;w')$, $v\in \WPRM_d(w')$, $u^t\in \WRM_{d-w_0}^\perp(w_0;w')$, $v^t\in \WPRM_d^\perp(w')$. We have
$$
\begin{aligned}
\ps{(u+v_\Lambda,v)}{(u^t,v^t-u_\Lambda^t)}&=\ps{v_\Lambda}{u^t}-\ps{v}{u^t_\Lambda}\\
&=\sum_{i=1}^{q-1}\ps{\Lambda(i)\star v}{u_i^t}-\ps{v}{\sum_{i=1}^{q-1} \Lambda(i)\star u_i^t}=0.
\end{aligned}
$$
This proves the statement, since every codeword of $\WPRM_d(w)$ is of the form $(u+v_\Lambda,v)$ by Theorem \ref{t:rec1weakcond}. 
\end{proof}

Note that Proposition \ref{p:dual} can also be applied in the context of subfield subcodes, substituting all the codes involved with their subfield subcodes, as long as the degree $d$ is as in Corollary \ref{c:ssc}. As before, the duals of the subfield subcodes of WRM codes are studied in \cite{galindostabilizer} as a particular 
of subfield subcodes of J-affine variety codes. 

For the rest of this section, we will need to consider orthogonality relations between the evaluation of monomials when evaluating in the affine space. Equivalently, we can study the sum of the evaluation of a monomial at every point of the affine space, and this can be understood with the following well-known result (a proof can be found in \cite[Lem. 4.2]{sanjoseSSCPRS}). 

\begin{lem}\label{lemasumafq}
Let $\gamma$ be a non-negative integer. We have the following:
$$
\sum_{z\in\F_q}z^\gamma=
\begin{cases}
0 &\text{ if } \gamma=0 \text{ or } \gamma>0 \text{ and } \gamma\not\equiv 0 \bmod q-1,\\
-1 &\text{ if } \gamma>0 \text{ and } \gamma\equiv 0 \bmod q-1.
\end{cases}
$$
\end{lem}

\begin{rem}\label{remafin}
Let $1 \leq \ell \leq m$ and $x_1^{\alpha_1}\cdots x_\ell^{\alpha_\ell}\in \fq[x_1,\dots,x_m]$. Then
$$
\sum_{Q\in \fq^\ell} x_1^{\alpha_1}\cdots x_\ell^{\alpha_\ell}(Q)=\left( \sum_{z\in \fq}x_1^{\alpha_1}(z) \right)\cdots \left( \sum_{z\in \fq}x_\ell^{\alpha_\ell}(z) \right).
$$
Thus, we can use Lemma \ref{lemasumafq} to obtain the result of this sum. In particular, 
$$
\sum_{Q\in \fq^\ell} x_1^{\alpha_1}\cdots x_\ell^{\alpha_\ell}(Q) \neq 0 \iff \forall i, \alpha_i > 0 \text{ and } q-1 \mid \alpha_i.$$
\end{rem}

This enables us to generalize results about the hull of PRM codes \cite{sanjoseHullsPRM,sanjose_hull_variation,kaplanHullsPRM,song_hulls_prm} to the weighted case in the particular case where $\gcd(w_i,q-1)=1$, for every $i \in \set{0,\dots,m}$. Recall that the (Euclidean) \emph{hull} is defined as $\text{Hull}(C)=C\cap C^\perp$. This object plays a role in several applications, such as determining the entanglement requirement of entanglement-assisted quantum error-correcting codes \cite{entanglement} built using the CSS construction \cite{galindoentanglement}.

Note that the hull, in general, depends on the choice of representatives. Restricting ourselves to the case where $\gcd(w_i,q-1)=1$, for $0\leq i \leq m-1$, we can fix the standard representatives of $\PP^m(\fq)$, i.e., the representatives obtained by considering for each point the representative with the leftmost nonzero entry equal to 1, as our chosen representatives for $\PP(w)(\fq)$. Indeed, this follows from Equation~\eqref{eq:1xA}, which can be applied recursively if the first $m$ weights are coprime with $q-1$. We will also call this \emph{set of standard representatives for $\PP(w)(\fq)$}. The following result generalizes \cite[Thm. 4.1]{kaplanHullsPRM}.

\begin{prop}
Let $w=(w_0,\dots,w_m) \in \NN^{m+1}$ such that $\gcd(w_i,q-1)=1$, for $i=0,\dots,m-1$. Consider the set of standard representatives of $\PP(w)(\fq)$. Let 
\[
D=\max\set{\sum_{i=0}^m \alpha_i :\sum_{i=0}^m \alpha_iw_i=d, \alpha_i \in \mathbb{N}}.
\]
If $2D<q-1$ (in particular, if $2d<\min(w)(q-1)$), then 
$$\dim \hull(\WPRM_d(w))=
\begin{cases}
    \dim \WPRM_d(w)-1 & \text{ if } w_m\mid d, \\
    \dim \WPRM_d(w) & \text{ otherwise}.
\end{cases}
$$
More precisely, the only monomial of degree $d$ whose evaluation does not lie in $\WPRM_d(w)^\perp$ is $x_m^{d/w_m}$, when $w_m\mid d$.
\end{prop}
\begin{proof}
Let $x^\alpha,x^\beta$ be two monomials of degree $d$. By the hypotheses and the choice of representatives, we have
$$
\ps{\ev(x^\alpha)}{\ev(x^\beta)}=\sum_{Q\in \mathcal{P}_w}x^{\alpha+\beta}(Q),
$$
where 
$\mathcal{P}_w$ is equal to the set of standard representatives of $\PP(w)$. By the hypotheses, $\deg(x^{\alpha+\beta})\leq  2D<q-1$, where we are considering the usual degree, not the weighted degree. The proof of \cite[Thm. 4.1]{kaplanHullsPRM} shows that this sum is equal to $0$, except when $x^{\alpha+\beta}=x_m^{2d}$, which gives the result. 
\end{proof}

\subsection{Representation as monomial codes}

In what follows, we seek a description for the duals of WPRM codes as evaluation codes. In particular, we will focus on describing them as monomial codes.

Let us first set some notations for monomials. We denote by $\M$ the set of all the monomials of $\fq[x_0,\dots,x_m]^w$. We also set
\[\M_d=\set{x^a=x_0^{a_0}\cdots x_m^{a_m} \in\M : \deg(x_0^{a_0}\cdots x_m^{a_m})=d}, \]
so that $\fq[x_0,\dots,x_m]^w_d=\Span \M_d$. Consider the degree lexicographic order with $x_0 < x_1 < \dots< x_m$. Let
\begin{equation}
    \Mon=\set{x_0^{a_0}\cdots x_m^{a_m} \in \M : \forall f \in  \calI(\mathbb{P}(w)(\fq)) \text{ homogeneous}, \: \ini(f) \nmid x_0^{a_0}\cdots x_m^{a_m} }
\end{equation}
and $\Mon_d=\Mon \cap \M_d$. We can write
\[\Mon=\bigsqcup_{d \geq 0} \Mon_d.\]
Then the quotient ring $\fq[x_0,\dots,x_m]^w /\calI(\mathbb{P}(w)(\fq))$ is generated by $\Mon$ as an $\fq$-vector space \cite[Thm. 15.3]{eisenbud}, and its homogeneous component of degree $d$ is generated by $\Mon_d$.

\begin{defn}\label{def:monomial_code}
    A code $C$ is said to be \emph{monomial of degree $d$} (in $\fq[x_0,\dots,x_m]^w$) if there exists a subset $\mathcal{M} \subseteq \M_d$ (or equivalently $\mathcal{M} \subseteq \Mon_d$) such that
    $C = \ev( \Span \mathcal{M})$. 
\end{defn}

The duals of $\PRM$ codes were previously computed by S\o rensen, whose result is recalled below \cite[Theorem~2]{sorensen}.

\begin{thm}\label{dualPRM}
Let $2\leq m$, $1\leq d\leq m(q-1)$ and $d^\perp=m(q-1)-d$. Then
\[
\PRM_d^\perp(m)=\begin{cases}
    \PRM_{d^\perp}(m) &\text{if } d\not\equiv 0\bmod q-1, \\
    \PRM_{d^\perp}(m)+\langle (1,\dots,1) \rangle &\text{if } d\equiv 0\bmod q-1.
\end{cases}
\]
\end{thm}

Thus, the dual codes of $\PRM$ codes are monomial, in the sense of Definition \ref{def:monomial_code}, if $d\not\equiv 0\bmod q-1$. Note that the result is true for any choice of representatives of $\PM(\fq)$, as long as we consider the same representatives for both $\PRM_d(m)$ and $\PRM_{d^\perp}(m)$. This is because given $f\in \fq[x_0,\dots,x_m]_d$, $f^{\perp}\in  \fq[x_0,\dots,x_m]_{d^\perp}$, we have that 
$$
\ps{\ev(f)}{\ev(f^\perp)}=\sum_{Q\in \PM(\fq)} (f f^\perp)(Q),
$$
where $ff^\perp \in \fq[x_0,\dots,x_m]_{m(q-1)}$. Thus, $ff^\perp(\lambda \cdot Q)=\lambda^{m(q-1)}ff^\perp(Q)=ff^\perp(Q)$, for any $\lambda \in \fq^*$, and the value of $\ps{\ev(f)}{\ev(f^\perp)}$ does not depend on the choice of representatives. A similar argument works for the vector $(1,\dots,1)$ when $d\equiv 0 \bmod q-1$.

In general, different representatives of $\PP(w)(\fq)$ give rise to monomially equivalent WPRM codes (recall Lemma \ref{l:change_of_representatives}). Since we analyze the duals of WPRM codes as monomial codes, we study $\ps{\ev(x^\alpha)}{\ev(x^\beta)}$, for $x^\alpha\in \Mon_d$, $x^\beta \in \Mon_{d^\star}$, for some degrees $d,d^\star >0$. Note that this is equivalent to studying the sums $\sum_{Q\in \PP(w)(\fq)}x^\gamma(Q)$, for any $x^\gamma\in \Mon_{d}\cdot \Mon_{d^\star}\subset \mathbb{M}_{d+d^\star}$. In some cases, these sums do not depend on the set of representatives chosen for $\PP(w)(\fq)$, as we show next.

\begin{lem}
The evaluation of a polynomial of degree $d$ in $\fq[x_0,\dots,x_m]^x_d$ at an $\fq$-point $Q\in \PP(w)(\fq)$ does not depend on the choice of representative if $\gcd(w_i, i \in \supp(Q))(q-1)$ divides $d$. In particular, if $\lcm(w)(q-1)$ divides $d$, the evaluation of a polynomial at any $\fq$-point does not depend on the choice of representative.
\end{lem}
\begin{proof}
By Lemma~\ref{l:lambdaQ}, $P$ and $Q$ are the representatives of the same $\fq$-point if and only if $Q = \lambda \cdot P$ for some $\lambda \in \Bar{\fq}$ such that $\lambda^{\gcd(w_i, i \in \supp(Q))} \in \fq$, i.e., $\lambda^{\gcd(w_i, i \in \supp(Q))(q-1)}=1$.  Then for any $f \in S_d^w$, $f(Q)=\lambda^d f(Q)=f(P)$, since $\gcd(w_i, i \in \supp(Q))(q-1) \mid d$.

The last assertion follows from the fact that $\gcd(w_i, i \in I)$ divides $\lcm(w)$ for any possible support $I \subseteq \{1,\dots,n\}$.
\end{proof}

\begin{lem}\label{l:bad_monomials}
Let $d, d^\star > 0$ such that $d+d^\star\equiv 0\bmod \gcd(d,\lcm(w))(q-1)$. Let $x^\alpha \in \mathbb{M}_d\cdot \mathbb{M}_{d^\star}$. Then
\begin{enumerate}[label=(\roman*)]
    \item $x^\alpha(Q)$ does not depend on the choice of representatives, and
    \item $\sum_{Q\in \PP(w)(\fq)}x^\alpha(Q)\neq 0$ if and only if $x^\alpha\equiv x_0^{q-1}\cdots x_m^{q-1}\bmod \calI(\mathbb{A}^{m+1}(\fq))$ (i.e., $\alpha_i>0$ and $\alpha_i\equiv 0\bmod q-1$, for $0\leq i \leq m$).
\end{enumerate}
\end{lem}
\begin{proof}

Let $Q$ be an $\fq$-point, and let $\lambda \cdot Q$, for some $\lambda \in \overline{\fq}$, be another representative of the same $\fq$-point.  From Lemma \ref{l:lambdaQ}, we have $\lambda^{\gcd(w_i,i\in \supp(Q))(q-1)}=1$.

Take $x^\alpha=x^\gamma x^\beta$ with $x^\gamma\in \mathbb{M}_d$, and $x^\beta\in \mathbb{M}_{d^\star}$. Then $x^\alpha(\lambda\cdot Q)=\lambda^{d+d^\star}x^\alpha(Q)$, and we have $x^\alpha(Q)=0$ if and only if $x^{\alpha}(\lambda \cdot Q)=0$.  If $x^\alpha(Q)\neq 0$, then $x^\gamma(Q)\neq0$ and  $\Gamma\subset \supp(Q)$. 

Write $x^\gamma= \prod_{i\in \Gamma}x^{\gamma_i}$ for some $\Gamma \subseteq \set{1,\dots,n}$ such that $\gamma_i>0$ for all $i\in \Gamma$. This implies that $d=\sum_{i\in \Gamma}\gamma_i w_i$, hence $\gcd(w_i,i\in \supp(Q))\mid  \gcd(w_i,i\in \Gamma) \mid d$. As $\gcd(w_i,i\in \supp(Q)) \mid \lcm(w)$, we deduce that $\gcd(w_i,i\in \supp(Q))\mid \gcd(d,\lcm(w))$ and then $\gcd(w_i,i\in \supp(Q))(q-1) \mid d+d^\star$, which implies that $\lambda^{d+d^\star}=1$. Thus, the value $x^\alpha(Q)$  does not depend on the choice of representatives.

Now, to prove $(ii)$, let $\overline{x^\alpha}$ be the reduced monomial (modulo $\calI(\mathbb{A}^{m+1}(\fq))$) such that $x^\alpha\equiv \overline{x^\alpha} \bmod \calI(\mathbb{A}^{m+1}(\fq)) $. Then 
$$
\sum_{Q\in \mathbb{A}^{m+1}\setminus \{0\}}x^\alpha(Q)=\sum_{Q\in \mathbb{A}^{m+1}\setminus \{0\}}\overline{x^\alpha}(Q)=(q-1)\sum_{Q\in \PP(w)(\fq)}\overline{x^\alpha}(Q)=(q-1)\sum_{Q\in \PP(w)(\fq)}x^\alpha(Q),
$$
where we have used that $\lambda^{d+d^\star}=1$. We finish the proof by considering Remark \ref{remafin}.
\end{proof}
\begin{rem}\label{r:bad_monomials}
If $\gcd(w_i,q-1)=1$, for $0\leq i \leq m$, and $d+d^\star\equiv 0 \bmod q-1$, the conclusion of the previous result also holds by Remarks \ref{r:lambda_in_fq} and \ref{remafin}.
\end{rem}

\begin{defn}
In what follows, if it exists, take $d^\star$ the smallest integer such that 
\begin{enumerate}
    \item $\WPRM_{d^\star}(w)=\fq^{p_m}$ (this means in particular that $\lcm(w) \mid d^\star$ by Lemma~\ref{l:degen}), \label{cond:1}
    \item $d+d^\star\equiv 0\bmod \gcd(d,\lcm(w))(q-1)$. \label{cond:2}
\end{enumerate}
Condition (\ref{cond:1}) ensures that $\WPRM_d(w)^\perp\subset \ev\left( \Mon_{d^\star} \right)$, and Condition (\ref{cond:2}) guarantees that the results do not depend on the choice of representatives. 
\end{defn}

\begin{rem}
If $d^\star $ exists, then $\gcd(q-1,\lcm(w))$ divides both $d+d^\star$ and $d^\star$, which implies that $\gcd(q-1,\lcm(w))\mid d$. Thus, if $\gcd(q-1,\lcm(w))\nmid d$, such a $d^\star$ cannot exist.
\end{rem}

Set $B(d,d^\star):=\set{x_0^{c_0}\cdots x_m^{c_m}\in \Mon_{d}\cdot \Mon_{d^\star} : \forall i \in \set{0,\dots,m},  c_i >0 \text{ and } q-1 \mid c_i }$. 

\begin{rem}\label{r:Bdd_non_empty}
By Lemma \ref{l:bad_monomials}, the set $B(d,d^\star)$ is non-empty, since otherwise all the evaluations of the monomials of degree $d^\star$ would be orthogonal to $\WPRM_d(w)$, and we would have $\fq^{p_m}=\WPRM_{d^\star}(w)\subset \WPRM_d^\perp(w)$, a contradiction. Moreover, given $x^\alpha\in \Mon_{d}$, $x^\beta \in \Mon_{d^\star}$, we have $\ev(x^\alpha)\cdot \ev(x^\beta)\neq 0$ if and only if $x^{\alpha+\beta}\in B(d,d^\star)$.  
\end{rem}

\begin{prop}\label{p:sub_dual}
Let $d>0$ and $d^\star$ as above. Then
  \[\Span \left(\ev \left(\Mon_{d^\star} \setminus \bigcup_{x^c \in B(d,d^\star)} \set{\overline{x^{c-a}}: x^a \in \Mon_{d},\; x^a \text{ divides } x^c}\right)\right) \subseteq \WPRM_d(w)^\perp. \]
\end{prop}
\begin{proof}
The result follows from Lemma \ref{l:bad_monomials} and Remark \ref{r:Bdd_non_empty}. 
\end{proof}

\begin{thm}\label{th:1_bad_monomial=>dual}
Let $d>0$ be such that $B(d,d^\star)=\{x^c\}$ (i.e., $\abs{B(d,d^\star)}=1$), where $d^\star$ is as above. Then
\begin{enumerate}[label=(\roman*)]
    \item every $x^a \in \Mon_{d}$ divides $x^c$ and
    \item $\WPRM_d(w)^\perp = \Span \left(\ev \left(\Mon_{d^\star} \setminus\set{\overline{x^{c-a}}: x^a \in \Mon_{d}}\right)\right)$.
\end{enumerate}
\end{thm}
\begin{proof}
Let us prove $(i)$ by contradiction. We assume that there exists a monomial $x^a \in \Mon_{d}$ such that $x^a$ does not divide $x^c$, and we will exhibit another monomial $x^{c'} \in B(d,d^\star)$.

By our assumption, the set $I:= \set{ i \in \set{0,\dots,m} : c_i < a_i}$ is non-empty. For every $i \in I$, we write $c_i-a_i=-\lambda_i (q-1)+r_i$ with $\lambda_i \geq 1$ and $1 \leq r_i \leq q-1$. Since $\gcd(w)=1$ and there are only finitely many gaps in a numerical semigroup, there exists a positive integer $\gamma$ such that $\gamma \lcm(w) - \sum_{i \in I} \lambda_i w_i$ lies in the semigroup generated by the weights $w_i$ for $i \in I$. In other words, there exists $\mu_i \geq \lambda_i\geq 1$ such that $\gamma \lcm(w)=\sum_{i \in I} \mu_i w_i$. Now, setting $b=(b_0,\dots,b_m)$ with
\begin{equation}\label{eq:def_b}
    b_i=\begin{cases}
    c_i-a_i+\mu_i(q-1) & \text{if } i \in I,\\
    c_i-a_i &\text{otherwise,}
    \end{cases}
\end{equation}
we get $x^{a+b}=x^c \prod_{i \in I} x_i^{(q-1)\mu_i} \equiv x^c \bmod \calI(\PP(w)(\fq))$. Since $\deg(x^a)=d$ and $\deg(x^c)=d+d^\star$, then $\deg(x^{b})=\deg(x^c)-\deg(x^a)+\deg\left( \prod_{i \in I} x_i^{(q-1)\mu_i}\right)=d^\star+\gamma\lcm(w)(q-1)$. Then $\tilde{d}:=\deg(x^b)$ also satisfies $\WPRM_{\tilde{d}}(w)=\fq^{p_m}$ (see, e.g., the proof of \cite[Lem. 2.7]{nardi_sanjose_conjecture}) and there is a bijection between $\Mon_{d^\star}$ and $\Mon_{\tilde{d}}$ modulo $\calI(\PP(w)(\fq))$. Thus, there exists a monomial $x^{\overline{b}} \in \Mon_{d^\star}$ such that $x^{b}\equiv x^{\overline{b}} \bmod \calI(\PP(w)(\fq))$. Then the monomial $x^{a+\overline{b}}$ is different from $x^c$ (otherwise, $x^{c-a}$ would have been a proper monomial) and it lies in $B(d,d^\star)$, which raises a contradiction.

\medskip

Now, let us prove (ii). Let $
\mathcal{T}:=\set{\overline{x^{c-a}}: x^a \in \Mon_{d}}.
$
By Item (ii) and Proposition \ref{p:sub_dual}, we have 
\[\WPRM_d(w)^\perp \supset \Span \left(\ev \left(\Mon_{d^\star} \setminus\mathcal{T}\right)\right). \]
Recall that $\WPRM_{d^\star}(w)=\fq^{p_m}$. Thus, it is enough to show $\abs{\mathcal{T}}=\dim \WPRM_d(w)=\abs{\Mon_d}$, since in that case $\dim\left(\ev\left(\Mon_{d^\star}\setminus \mathcal{T}\right)\right)=\dim \WPRM_d(w)^\perp$. Assume that we had $\overline{x^{c-a}}\equiv \overline{x^{c-b}}\bmod \calI(\PP(w)(\fq))$, with $x^a,x^b\in \Mon_d$. If we multiply by $x^a$, since $\overline{x^{c-a}}x^a\equiv x^c\bmod \calI(\PP(w)(\fq))$, we get $x^c\equiv \overline{x^{c-b}}x^a\bmod \calI(\PP(w)(\fq))$. By Theorem \ref{t:binomials_vanishing_ideal}, and Lemma \ref{l:bad_monomials}, we have $\overline{x^{c-b}}x^a\in B(d,d^\star)$ (note $\overline{x^{c-b}}\in \Mon_{d^\star}, x^a\in \Mon_d$). By assumption, we get $\overline{x^{c-b}}x^a=x^c$, i.e., $\overline{x^{c-b}}=x^{c-a}$. Similarly, $\overline{x^{c-a}}=x^{c-b}$. Then both $x^{c-a},x^{c-b}$ are reduced, and $x^{c-a}=x^{c-b}$, which implies $x^a=x^b$.
\end{proof}

\begin{rem}
If, instead of Equation \eqref{eq:def_b}, we had set $b'=(b_0',\dots,b_m')$ with
\[b'_i=\begin{cases}
    c_i-a_i+\lambda_i(q-1) & \text{if } i \in I,\\
    c_i-a_i &\text{otherwise,}
\end{cases}\]
we would also have obtained $x^{a+b'}=x^c \prod_{i \in I} x_i^{(q-1)\lambda_i} \equiv x^c \bmod \calI(\PP(w)(\fq))$ but the monomial $x^{b'}$ would have degree $d^\star+(q-1) \sum_{i \in I} \lambda_i w_i$. If $\lcm(w) \nmid \sum_{i \in I} \lambda_i w_i$, we cannot ensure the existence of a \emph{monomial} in $\Mon_{d^\star}$ that is equal to $x^{b'}$ modulo $\calI(\PP(w)(\fq))$. 
\end{rem}

Theorem~\ref{th:1_bad_monomial=>dual} states that if $B(d,d^\star)$ consists in only one monomial, then the dual of $\WPRM_d(w)$ is monomial of degree $d$. However, this condition is only necessary, as illustrated by $\PRM$ codes (i.e., $w=(1,\dots,1))$. For more general weights, with the extra condition $d<\min(w)(q-1)$, we can be more precise about $B(d,d^\star)$ and $\WPRM_d^\perp(w)$.

\begin{cor}\label{cor:only_1_bad_monomial}
Let $0<d<\min(w)(q-1)$ and $d^*$ as above. Let $x^c$ be the only monomial in $B(d,d^\star)\cap \overline{\mathbb{M}}_{d+d^\star}$. Then $B(d,d^\star)=\{x^c\}$ and
\[\WPRM_d(w)^\perp = \Span \left(\ev \left(\Mon_{d^\star} \setminus\set{\overline{x^{c-a}}: x^a \in \M_{d}}\right)\right). \]
\end{cor}
\begin{proof}
By Remark \ref{r:Bdd_non_empty}, $B(d,d^\star)$ is not empty. Assume that $x^{c'}\in B(d,d^\star)$, with $c\neq c'$. Let $x^\beta \in \overline{\mathbb{M}}_d$. Since $d<\min(w)(q-1)$, we have $\beta_i<q-1$, for $0\leq i \leq m$. Then $x^{c-\beta}-x^{c'-\beta}\in \calI(\PP(w)(\fq))$ by Theorem \ref{t:binomials_vanishing_ideal}. Because $x^c$ is reduced, we have $x^c<x^{c'}$, and $x^{c-\beta}<x^{c'-\beta}$. This is true for any $x^\beta \in \overline{\mathbb{M}}_d$, and we reach a contradiction, since this implies $x^{c'-\beta}$ is not reduced for any $x^\beta \in \overline{\mathbb{M}}_d$, which would entail that $x^{c'}\not \in \Mon_{d}\cdot \Mon_{d^\star}$.
The result about the dual follows from Theorem~\ref{th:1_bad_monomial=>dual}, noticing that when $d<\min(w)(q-1)$, we have no polynomial of degree $d$ in $\calI(\PP(w)(\fq))$, hence $\Mon_d=\M_d$.
\end{proof}

The following example illustrates that the condition $d < \min(w)(q-1)$ of Corollary~\ref{cor:only_1_bad_monomial} to have $\size{B(d,d^\star)}=1$ is only necessary.

\begin{ex}
    Set $q=5$ and $w=(2,5,7)$. Then $\gcd(q-1,\lcm(w))=2$.

For $(d,d^\star)=\set{(2,350),(4,420),(6,210),(10,350),(12,420),(14,210)}$, we checked with $\textsc{Magma}$ \cite{magma} that $\size{B(d,d^\star)}=1$, so the dual of $\WPRM_d(2,5,7)$ is given by Theorem~\ref{th:1_bad_monomial=>dual}. 

For $d=8,\:16$, we take $d^\star=280$ and then $B(d,d^\star)=\set{x_0^{260}x_1^4x_2^4,\:x_0^8 x_1^{84} x_2^{20}}$. With $\textsc{Magma}$, we checked that that the dual codes $\WPRM_d(2,5,7)^\perp$ admits one extra generator outside the monomial part described in Proposition~\ref{p:sub_dual}, that is the evaluation of the binomial $x_0^{116}x_1^4x_2^4-x_1^{28}x_2^{20}$.

\end{ex}

\section{Schur products of WPRM codes}\label{s:schur}

In this section, we investigate the Schur product of two WPRM codes. We connect this question with a known problem about polytopes, and show a particular case in which the Schur product of two WPRM codes is also a WPRM code. 

Set $w=(w_0,\dots,w_m)\in \NN_{\geq 1}^{m+1}$. For any degrees $d_1, \: d_2$, we have
\[\WPRM_{d_1}(w) \star \WPRM_{d_2}(w) \subseteq  \WPRM_{d_1+d_2}(w).\]
In this section, we investigate the necessary and sufficient conditions to get equality. If $d_1+d_2 < \min(w)q$, this is equivalent to finding conditions so that
\begin{equation}\label{eq:pd_Sd}
     \M_{d_1}\cdot \M_{d_2} = \M_{d_1+d_2}
\end{equation} If $d_1+d_2$ is larger, then the previous condition is sufficient to get the equality for the associated codes. This property does not come for free, as illustrated by the next example.

\begin{ex}\label{ex:S1S1neS2}
For $w=(1,1,2)$, the property is not fulfilled for $d_1=d_2=1$, as $\M_1=\set{x_0,x_1}$ but $x_2 \in \M_2$.
\end{ex}

Leveraging the combinatorics underlying toric geometry, we can reformulate this question in terms of polytopes. A degree $d$ defines an $m$-dimensional simplex $P_d$ as follows (see \cite[\textsection 1.7]{rossi2011weighted} for details).
\begin{itemize}
    \item If $w_0=1$, $P_d$ is the rectangular simplex defined as the convex hull of the origin and the points $\frac{d}{w_i}e_i$ in $\R^m$, where $e_i$ denotes the points whose coordinates are all zeros, but the $i^{th}$ being one. In this case, the integral points of $P_d$ are in one-to-one correspondence with the monomials of degree $d$: a point $(a_1,\dots,a_m) \in P_d \cap \Z^m$ corresponds to the monomial $x_0^{a_0}x_1^{a_1}\cdots w_m^{a_m}$ where $a_0=d-\sum_{i=1}^m  w_i a_i$.
    \item If $w_0 \geq 2$, $P_d$ is the \emph{intersection} in $\R^{m+1}$ of the $(m+1)$-simplex whose vertices are the origin and the $\frac{d}{w_i}e_i$, with the hyperplane defined by $\sum_{i=0}^m x_i w_i = d$. In this case, the integral points of $P_d$ are precisely the exponents of monomials of degree $d$. It is also possible to define $P_d$ directly in $\R^m$, by computing its normal fan using the transition matrix to the Hermite Normal Form of the vector $(w_0,\dots,w_m)$. 
\end{itemize}
Then Equation \eqref{eq:pd_Sd} holds if and only if 
\begin{equation}\label{eq:Pd1Pd2}
(P_{d_1} \cap \Z^m) + (P_{d_2}\cap \Z^m) = P_{d_1+d_2} \cap \Z^m,     
\end{equation}
which matches the so-called integer decomposition property of polytopes (see \cite{haase2017_pair_polytopes}).

\begin{defn}
An $m$-dimensional polytope $P$ is said to have the \emph{integer decomposition property} (IDP) (or to be \emph{normal}), if for all $\ell \in \NN$ and all $z \in  (\ell P) \cap  \Z^m$, there exist $x_1,\dots,x_\ell \in P \cap \Z^m$ such that $z = x_1 + \dots + x_\ell$.

A pair of $m$-dimensional polytopes $(P,Q)$ is said to have the \emph{integer decomposition property} if $(P \cap \Z^m) + (Q\cap \Z^m) = (P+Q) \cap \Z^m$.     
\end{defn}

Independently of the value of $w_0$, the simplices $P_d$ are all scalar multiples of a same simplex whose vertices lies in $\frac{1}{\lcm(w)}\Z^m$. In particular, the simplex $P_d$ is integral (i.e., all its vertices are have integer coordinates) if and only if $\lcm(w)$ divides $d$.

Let us now focus on the case where $d_1$ and $d_2$ are divisible by $\lcm(w)$. Otherwise, Example \ref{ex:S1S1neS2} shows that the desired property is likely to fail. Let us set $\delta=\lcm(w)$ and write $d_i=\ell_i \delta$ for $i=1,2$. In this case, it is easy to check that if the simplex $P_\delta$ has the IDP, then $P_{d_i} \cap \Z^m=(\ell_i P_\delta) \cap \Z^m$ for $i=1,2$ and the equality in Equation \eqref{eq:Pd1Pd2} holds.

\begin{prop}\label{prop:schur}
    Set $w=(w_0,\dots,w_m)$ with $m \geq 1$ and $\delta=\lcm(w)$. If the lattice simplex $P_{\delta}$ had the IDP, then for any degrees $d_1, \: d_2$ divisible by $\lcm(w)$, we have
    \[\WPRM_{d_1}(w) \star \WPRM_{d_2}(w) =  \WPRM_{d_1+d_2}(w).\]
\end{prop}

Every one or two dimensional lattice polytope has the IDP \cite[Corollary~2.54]{bruns2009polytopes}. In dimension 3, some simplices do not satisfy the IDP (see Example~\ref{eq:Ogata}). Characterizing integral simplices with the IDP is an active research topic (e.g., see \cite{braun2018detecting} for reflexive simplices and \cite{adeyemo2025embeddingsweightedprojectivespaces} for the rectangular ones, i.e., corresponding to $\PP(w)$ with $w_0=1$).

However, for any $m$-dimensional integral polytope $P$, $\ell P$ has the IDP for every $\ell \geq m-1$ \cite[Proposition~1.1]{ogata2002generators}. Moreover, repeating weights does not impact the IDP of $P_\delta$ \cite[Proposition 3.1]{adeyemo2025embeddingsweightedprojectivespaces}. Gathering these two results, we get the following proposition.

\begin{prop}\label{prop:schur2}
    Set $w=(w_0,\dots,w_m)$ with $m \geq 1$. Let $s=\size{\set{w_0,\dots,w_m}}$ be the number of different weights. For any degrees $d_1, \: d_2$ divisible by $\max(1,s-2)\lcm(w)$, we have
    \[\WPRM_{d_1}(w) \star \WPRM_{d_2}(w) =  \WPRM_{d_1+d_2}(w).\]
\end{prop}
With Proposition~\ref{prop:schur2} and $w=(1,\dots,1)$, we recover the well-known fact that 
\[\PRM_{d_1}(m) \star \PRM_{d_2}(m) =  \PRM_{d_1+d_2}(m)\]
for any degrees $d_1,\:d_2 \geq 0$.

\begin{ex}\label{eq:Ogata} The 3-dimensional rectangular simplex associated to $w=(1,6,10,15)$ does not have the IDP. This famous counterexample is due to Ogata \cite[p.522]{ogata2005k}. In the formalism of codes, it means that for $q$ large enough,
\[\WPRM_{30}(w) \star \WPRM_{30}(w) \neq  \WPRM_{60}(w)\]
because $x_0x_1^4x_2^2x_3$ has degree $60$ but it cannot be written as the product of two monomials of degree $30$. One of these monomials would contain $x_3$ and we would have to write $15=i+6j+10\ell$ with $i\in\set{0,1}$, which is impossible. It is the only monomial with such a behavior, as $\dim(\WPRM_{60}(w))=81=1+\dim\left(\WPRM_{30}(1,6,10,15) \star \WPRM_{30}(1,6,10,15) \right)$.
\end{ex}

\section*{Acknowledgments}
The first author is supported  by the French National Research Agency through ANR \textit{Barracuda} (ANR-21-CE39-0009) and the French government \textit{Investissements d’Avenir} program ANR-11-LABX-0020-01. The second author was supported in part by the Grant DMS-2401558 funded by the National Science Foundation, Grant PID2022-137283NB-C22 funded by MICIU/AEI/10.13039/501100011033 and by ERDF/EU, and by the Commonwealth Cyber Initiative.


\end{document}

%% file: macros.tex
\DeclareMathOperator{\supp}{supp} 
\DeclareMathOperator{\ini}{in} 
\DeclareMathOperator{\ev}{ev} 

\DeclareMathOperator{\PRS}{PRS}
\DeclareMathOperator{\Span}{Span}
\DeclareMathOperator{\PRM}{PRM}
\DeclareMathOperator{\RM}{RM}

\DeclareMathOperator{\wt}{wt}

\DeclareMathOperator{\RS}{RS}
\DeclareMathOperator{\WPRM}{WPRM}
\DeclareMathOperator{\WRM}{WRM}

\DeclareMathOperator{\WPRS}{WPRS}
\DeclareMathOperator{\den}{den} 
\DeclareMathOperator{\lcm}{lcm} 
\DeclareMathOperator{\Cone}{Cone}
\DeclareMathOperator{\ord}{ord}

\DeclareMathOperator{\hull}{Hull}

\newcommand{\F}{{\mathbb{F}}}
\newcommand{\fq}{\mathbb{F}_q}

\newcommand{\PP}{{\mathbb{P}}}
\newcommand{\PM}{{\mathbb{P}^{m}}}

\newcommand{\NN}{{\mathbb{N}}}
\newcommand{\Z}{{\mathbb{Z}}}
\newcommand{\A}{{\mathbb{A}}}

\newcommand{\calI}{{\mathcal{I}}}
\newcommand{\calP}{{\mathcal{P}}}

\newcommand{\set}[1]{\left\{#1\right\}}
\newcommand{\size}[1]{\left|#1\right|}
\newcommand{\gen}[1]{\left\langle #1 \right\rangle}
\newcommand{\Floorfrac}[2]{\left\lfloor \frac{#1}{#2} \right\rfloor}
\newcommand{\ps}[2]{\left\langle#1,#2\right\rangle}
\newcommand{\R}{{\mathbb{R}}}

\renewcommand{\bar}{\overline}
\newcommand{\Mon}{\bar{\mathbb{M}}}

\newcommand{\M}{\mathbb{M}}